\documentclass [journal,onecolumn,11pt]{IEEEtran}
\usepackage{amsfonts,amsmath,amssymb}
\usepackage{indentfirst, setspace}
\usepackage{url,float}
\usepackage{color}
\usepackage[a4paper,ignoreall]{geometry}
\usepackage{longtable}
\usepackage{graphicx}
\usepackage[a4paper,ignoreall]{geometry}
\usepackage{multicol}
\usepackage{stfloats}
\usepackage{enumerate}
\usepackage{cite}
\usepackage{amsthm}
\usepackage{mathrsfs}
\usepackage{multirow}
\usepackage{amssymb}
\usepackage[all]{xy}
\usepackage[overload]{empheq}
\usepackage[square, comma, sort&compress, numbers]{natbib}
\usepackage{float}
\usepackage{booktabs}
\usepackage{diagbox}
\usepackage{multirow}
\usepackage{makecell}
\usepackage{threeparttable}
\usepackage{bookmark}
\allowdisplaybreaks[4]

\newtheorem{Thm}{Theorem}[section]

\newtheorem{Prop}[Thm]{Proposition}

\newtheorem{Lem}[Thm]{Lemma}

\newtheorem{example}{Example}

\newtheorem{Rem}{Remark}

\newcommand{\tr}{{\rm Tr}}

\newcommand{\F}{{\mathbb F}}
\newcommand{\Rm}[1]{\uppercase\expandafter{\romannumeral #1}}

\newcommand{\com}[1]{\iffalse #1 \fi}

\makeatletter
\newcommand{\figcaption}{\def\@captype{figure}\caption}
\newcommand{\tabcaption}{\def\@captype{table}\caption}
\makeatother

\begin{document}

    \title{Balanced Boolean functions with few-valued Walsh spectra parameterized by $P(x^2+x)$}
    \author{
    {Qiancheng Zhang, Kangquan Li, Longjiang Qu}
    \thanks{Qiancheng Zhang, Kangquan Li, and Longjiang Qu are with the College of Science, National University of Defense Technology, Changsha, 410073, China. E-mail:  zhangqiancheng20@nudt.edu.cn, likangquan11@nudt.edu.cn, ljqu\_happy@hotmail.com. This work is supported by the National Key Research and Development Program of China under Grant 2024YFA1013000, and the National Natural Science Foundation of China (NSFC) under Grant 62202476. {\emph{(Corresponding author: Kangquan Li)}}}
    }
    
    \maketitle{}

    \begin{abstract} 
    Boolean functions with few-valued spectra have wide applications in cryptography, coding theory, sequence designs, etc. In this paper, we further study the parametric construction approach to obtain balanced Boolean functions using $2$-to-$1$ mappings of the form  $P(x^2+x)$, where $P$ denotes carefully selected permutation polynomials. The key contributions of this work are twofold: (1) We establish a new family of four-valued spectrum Boolean functions. This family includes Boolean functions with good cryptographic properties, e.g., the same nonlinearity as semi-bent functions, the maximal algebraic degree, and the optimal algebraic immunity for dimensions $n \leq 14$. (2) We derive seven distinct classes of plateaued functions, including four infinite families of semi-bent functions and a class of near-bent functions.
    \end{abstract}


    \section{Introduction}
    The Walsh transform of Boolean functions is a fundamental tool to characterize important information about cryptography. Recently, a significant demand for Boolean functions with few-valued Walsh spectra has emerged, hereinafter called few-valued spectrum functions, see \cite{li2020class,xu2017several,li2013walsh,luo2008weight,hodvzic2020characterization,su2022constructions,maitra2002cryptographically}. This arises because a few-valued Walsh spectrum usually admits various desirable properties to resist fast correlation attacks \cite{meier1988fast} and linear cryptanalysis \cite{matsui1993linear}. Moreover, few-valued spectrum functions are also closely related to cyclic codes \cite{luo2008cyclic,li2022class} and cross-correlations of $m$-sequences \cite{helleseth2005new,dobbertin2006niho,canteaut2000binary,luo2016binary,niho1972multi}. 

    All two-valued spectrum Boolean functions have been comprehensively determined as affine functions, bent functions, and their modification at zero by Tu et al. \cite{tu2011boolean}. Bent functions, firstly introduced by Rothaus in \cite{rothaus1976bent}, achieve the maximum nonlinearity to resist linear attacks. The inherent unbalancedness, however, limits their direct application in cryptography. New classes of Boolean functions that allow both good nonlinearity and balancedness are needed. Later, Chee et al. \cite{chee1995semi} proposed the semi-bent functions as a natural extension of bent functions, Zheng and Zhang further generalized them as the plateaued functions in \cite{zheng1999plateaued}.

    The so-called $r$-plateaued functions are three-valued spectrum functions whose Walsh spectra take values in $\{0,\pm 2^{\frac{n+r}{2}}\}$ for some fixed $0< r \leq n$. The subclass of $1$-plateaued is the near-bent function, and $2$-plateaued is the semi-bent function. A robust set of properties is exhibited in plateaued functions, such as high nonlinearity and good autocorrelation characteristics, making them ideal in cryptography \cite{mesnager2014semi}. Ongoing work has been done on their characterization and design over the past twenty years \cite{hodvzic2019designing,hodvzic2020general,carlet2003plateaued,mesnager2014semi,carlet2015boolean,hodvzic2019generic,cusick2017highly,logachev2010recursive,cao2016further}. In \cite{carlet2003plateaued}, Carlet investigated a characterization of plateaued functions by extending the notion of covering sequence. More results were given in \cite{carlet2015boolean} by their derivatives, autocorrelation functions, and the power moments of the Walsh transform, while this method still deserves further work. Alternative primary constructions of plateaued functions are mainly achieved by specifying algebraic normal form \cite{carlet2003plateaued,cusick2017highly,logachev2010recursive}, or the reversed design to specify values in the spectral domain \cite{hodvzic2019designing}. Despite yielding significant positive outcomes, previous methods seem to lack simplicity and require a substantial amount of detail. 

    In contrast to plateaued functions, the known four-valued spectrum functions are relatively scarce, and existing approaches mainly derive two categories. One features a spectral set of $\{0,\pm a,b\}$, and the other takes values in $\{\pm a,\pm b\}$, where the absolute values of nonzero integers $a$ and $b$ differ. The first category comprises power functions in trace form with Niho exponents \cite{helleseth2005new,xie2009class,dobbertin1998one,li2013walsh,niho1972multi,dobbertin2006niho}, while the second category is derived from bent functions either by complementing values at two specific points \cite{sun2015boolean} or via the indirect sum approach \cite{sun2017several}.  

    Very recently, motivated by \cite{carlet2024Parameterization}, we \cite{qu2025parametric} delved into the parametric construction approach of balanced Boolean functions, where their supports coincide with the image sets of $2$-to-$1$ mappings over $\F_{2^n}$. In particular, the relationship between $2$-to-$1$ mappings, especially ones of the form $P(x^2+x)$, and the Walsh transform of the corresponding balanced Boolean functions was established. In this paper, we further study this approach to construct balanced Boolean functions with few-valued spectra. It should be noted that various sparse permutation polynomials (see e.g. \cite{li2019survey}) constructed by researchers over the past decade can all be effectively employed in our proposed construction framework.
    This approach enables us to start with a relatively simple permutation and get a complicated Boolean function distinct from known classes. In addition, the cryptographic criteria of the new function can be linked to properties of the permutation. This allows the parallel utilization of certain techniques for studying permutations to calculate those criteria, thereby providing great convenience. The main contributions of this paper include: (\Rm{1}) a novel class of four-valued spectrum functions from the permutation of the form $cx+\tr_m^{2m}(x^{\frac{1}{5}(2^m-1)+1})$, which includes Boolean functions with good cryptographic properties, e.g., the same high nonlinearity as semi-bent functions, the upper bound of algebraic degree, and the optimal algebraic immunity for dimensions $n \leq 14$; (\Rm{2}) a class of plateaued functions from permutations of the form $x^{3 \cdot 2^m }+ax^{2\cdot 2^m+1}+bx^{2^m+2}+x^3$, including a subclass of semi-bent functions; (\Rm{3}) three classes of plateaued functions from the so-called Dembowski-Ostrom (DO) permutations, including a subclass of near-bent functions; and (\Rm{4}) three classes of plateaued functions from permutations of the form $cx+\tr_{m}^{km}(x^s)$ where the binary weight of $s$ equals $2$, yielding three semi-bent subclasses.

    The outline of this paper is as follows. Section \ref{preliminaries} establishes the relationship between $2$-to-$1$ mappings and parameterized Boolean functions through the Walsh transform. In Section \ref{four-valued}, we present a new class of four-valued spectrum functions, as well as its spectral value distribution. In Section \ref{plateaued functions}, several quadratic permutations with sparse terms are employed to construct plateaued functions. Finally, Section \ref{conclusion} concludes this paper.

    \section{Preliminaries}\label{preliminaries}
    The notation used in this paper is summarized below.
    \begin{itemize}
        \item $\F_{2^n}$: the finite field with $2^n$ elements;
        \item $\F_{2^n}^*$: the multiplicative subgroup $\F_{2^n} \backslash \{0\}$;
        \item $\alpha$: a primitive element of $\F_{2^n}$;
        \item $\F_{2^n}[x]$: the polynomial ring over $\F_{2^n}$;
        \item $\#E$: the cardinality of the set $E$;
        \item $E \backslash D$: the set including elements that belong to $E$ but not $D$;
        \item $\bar{x}$: the $2^m$-th power of the element $x \in \F_{2^n}$ with $n=2m$;
        \item $U_m$: the unit circle $\{z \in \F_{2^n} \mid z\bar{z}=z^{2^m+1}=1\}$ with $n=2m$;
        \item $f,\mathcal{NL}(f),\deg(f)$: a Boolean function, along with its nonlinearity, algebraic degree;  
        \item $\tr_m^n(\cdot)$: the trace function from $\F_{2^n}$ to $\F_{2^m}$ with $m|n$ defined as 
        \begin{eqnarray*}
            \tr_m^n(x)~:~\F_{2^n} & \to & \F_{2^m},\\
            x & \mapsto & x+x^{2^m}+\cdots+x^{2^{\left(\frac{n}{m}-1\right)m}}.
        \end{eqnarray*}
        When $m=1$, it becomes the absolute trace function and is denoted by $\tr_n(\cdot)$ for simplicity. 
    \end{itemize}
    
    \subsection{Walsh transform and Weil sums}
    We refer to the mappings $F$ from $\F_{2^n}$ to $\F_{2^m}$ as $(n,m)$-functions, with Boolean functions $f$ specifically denoting those from $\F_{2^n}$ to $\F_2$. 
    Boolean functions possess many representations, a classical one considers them as a univariate polynomial of the form
    \begin{eqnarray}\label{unibool}
        f(x)=\sum_{i=0}^{2^n-1}{a_ix^i},
    \end{eqnarray}
    where $a_0,a_{2^n-1} \in \F_{2}$, $a_i \in \F_{2^n}$ and $a_{2i \pmod{2^n-1}} = a_i^2$ holds with $1 \leq i <2^n-2$. 
    
    For a Boolean function $f$ over $\F_{2^n}$, its Walsh transform at $\gamma \in \F_{2^n}$ is defined as 
    \[
        W_f(\gamma)=\sum_{x \in \F_{2^n}}{(-1)^{f(x)+\tr_n(\gamma x)}},
    \]
    and the multiset $\{*~W_f(\gamma)~|~\gamma \in \F_{2^n}~*\}$ is called the Walsh spectrum of $f$. Suppose that $f$ is a $t$-valued spectrum function and define 
    $M_i=\# \{\gamma \in \F_{2^n} \mid W_f(\gamma)=v_i\}$ for $0 < i \leq t$. Then we have three pivotal constraints on the Walsh spectrum \cite{li2013walsh}
    \begin{equation}\label{wr}
        \begin{cases*}
            \sum_{i=1}^{t}{M_i}=2^n,\\
            \sum_{i=1}^{t}{M_i v_i}=2^n(-1)^{f(0)},\\
            \sum_{i=1}^{t}{M_i v_i^2}=2^{2n}.
        \end{cases*}
    \end{equation}
 
    For a function $F$ from $\F_{2^n}$ to itself, its Walsh transform at $(a,b) \in \F_{2^n} \times \F_{2^n}$ is given by 
    \[
        W_F(a,b)=\sum_{x \in \F_{2^n}}{(-1)^{\tr_n(b F(x) + a x)}}.
    \]
    A Weil sum over $\F_{2^n}$ is an exponential sum of the form \cite{lidl1997finite}
    \[
        S(F,n)=\sum_{x \in \F_{2^n}}{\chi_n(F(x))}, 
    \]
    where $\chi_n(x)=(-1)^{\tr_n(x)}$ denotes the canonical additive character on $\F_{2^n}$. Then it is obvious that 
    \[
        W_F(a,b)=S(bF(x)+ax,n).
    \]
    
    In terms of a general function, explicitly calculating its Weil sum is often challenging, except for some special cases. Let $t$ be a positive integer. The function $F$ is called a $t$-to-$1$ mapping of $\F_{2^n}$ if it has either $t$ preimages or none for any element in $\F_{2^n}$. We refer to $F$ as a permutation (resp. $2$-to-$1$ mapping) when $t=1$ (resp. $t=2$).
    \begin{Lem}\label{org}\cite{lidl1997finite}
        Let $n$ be a positive integer and $P \in \F_{2^n}[x]$ be a permutation. Then for any $a \in \F_{2^n}$, the following formula holds
        \[
            S(aP(x),n)=
            \begin{cases*}
                0,  \text{~~if~} a \ne 0, \\
                2^{n},  \text{~if~} a=0.
            \end{cases*}
        \]
    \end{Lem}
    \begin{Lem}\cite{carlitz1979explicit}\label{Three_sum}
        Let $n$ be an odd integer. For any $a \in\F_{2^n}$, the following formula holds 
        \[
            S(x^3+ax,n)=
            \begin{cases}
                \left(\frac{2}{n}\right) \chi_n{\left(\theta^3+\theta\right)} 2^{\frac{n+1}{2}}, & ~\text{if}~ x^{4}+x=(a+1)^2~\text{has a solution $\theta \in \F_{2^n}$},\\
                0,&~\text{otherwise},
            \end{cases}
        \]
        where $\left(\frac{2}{n}\right)$ is the Legendre symbol.
    \end{Lem}

    \subsection{Equivalent relations and cryptographic criteria of Boolean functions}
    There are three main equivalent relationships among functions in finite fields. Two $(n,m)$-functions $F$ and $F^{\prime}$ are called \cite{carlet2020boolean}
    \begin{itemize}
        \item affine equivalent: if $F^{\prime}=A_1 \circ F \circ A_2$, where $A_1$ is an affine automorphism of $\F_{2^m}$ and $A_2$ is an affine automorphism of $\F_{2^n}$.
        \item extended affine equivalent (EA-equivalent for short): if $F^{\prime}=A_1 \circ F \circ A_2 + A$, where $A_1$ is an affine automorphism of $\F_{2^m}$, $A_2$ is an affine automorphism of $\F_{2^n}$, and $A$ is an affine $(n, m)$-function.
        \item Carlet-Charpin-Zinoviev equivalent (CCZ-equivalent for short): if the graphs of $F$ and $F^{\prime}$ are affine equivalent, that is $A(\mathcal{G}(F))=\mathcal{G}(F^{\prime})$ for some affine automorphism $A$ on $\F_{2^n} \times \F_{2^m}$, where $\mathcal{G}(F)=\{ (x, F(x)) \mid x \in \F_{2^n} \}$ and $\mathcal{G}(F^{\prime})=\{ (x, F^{\prime}(x)) \mid x \in \F_{2^n} \}$.
    \end{itemize}
    
    Affine equivalence is a strict particular case of EA equivalence, which in turn constitutes a specialized form of CCZ equivalence. It should be noted that CCZ-equivalence and EA-equivalence are consistent for Boolean functions. 
    
    The main cryptographic criteria of Boolean functions defined on $\F_{2^n}$ are as follows:
    \begin{itemize}
        \item \textit{Balancedness:} A Boolean function $f$ is called balanced if and only if $W_f(0)=0$.
        \item \textit{Nonlinearity:} The nonlinearity of a Boolean function $f$ is defined as
        \[
            \mathcal{NL}(f)=2^{n-1}-\frac{1}{2}\max_{a \in \F_{2^n}}|W_f(a)|.
        \]
        \item \textit{Algebraic degree:} For a Boolean function in the form as (\ref{unibool}), the algebraic degree of $f$ is given by
        \[
            \deg(f)=\max{\left\{ w_2(i) \mid  a_i \ne 0 \right\}},
        \]
        where $w_2(i)$ denotes the binary weight of $i$. When $\deg(f)$ is not greater than $1$ or $2$, $f$ is called an affine or quadratic function, respectively.
        \item \textit{Algebraic immunity:} For a Boolean function $f$,
        the algebraic immunity $\mathrm{AI}(f)$ is defined as the minimum algebraic degree of a nonzero annihilator $h$, where $h$ is a Boolean function such that $fh=0$ or $(f+1)h=0$.
    \end{itemize}

    In the equivalence sense, the algebraic degree and nonlinearity are EA invariants, while the balancedness and algebraic immunity are affine invariants. To withstand known attacks, Boolean functions employed in stream ciphers should satisfy the main cryptographic properties, including balancedness, high nonlinearity, and high algebraic degree.
    
    \subsection{Boolean functions from $2$-to-$1$ mappings}
    For any Boolean function $f$ over $\F_{2^n}$,   its support is defined as the set $E=\{ x \in \F_{2^n} \mid f(x)=1 \}$. Conversely, an arbitrary $f$ can be uniquely represented by the support set $E$ as 
    \[
        f(x)=
        \begin{cases*}
            1,  \text{~if~} x \in E, \\
            0,  \text{~if~} x \notin E.
        \end{cases*}
    \]
    Based on this representation, Carlet chose the image set of certain injective $(n,m)$-functions as the support set and derived excellent Boolean functions in \cite{carlet2024Parameterization}. Another natural idea \cite{qu2025parametric} is to employ a $2$-to-$1$ mapping instead, which ensures that the support set comprises half of the elements in the domain, leading to a balanced Boolean function. 
    
    Let $F$ be a $2$-to-$1$ mapping from $\F_{2^n}$ to itself. Given the $2$-to-$1$ property of $F$, one can derive a {balanced Boolean function $f$}, whose support equals the image set of $F$, i.e.,
    \[
        f(x)=
        \begin{cases*}
            1,  \text{~if~} x \in \text{Im}(F), \\
            0,  \text{~if~} x \notin \text{Im}(F).
        \end{cases*}
    \]
    
    The above method is called the parametric construction approach, and we say that $f$ is parameterized by $F$. 
    {Clearly, the induced mapping of $x^2+x$ is $2$-to-$1$ over $\F_{2^n}$, and the composition $P(x^2+x)$ preserves this property for any permutation $P(x) \in \F_{2^n}[x]$.
In this paper, we further study the parametric construction approach using $2$-to-$1$ mappings of the form $P(x^2+x)$, where the corresponding balanced Boolean function is denoted by $f_P$.
 In \cite{qu2025parametric}, we established a relationship between $P(x^2+x)$ and the corresponding $f_{P}$ as follows. 
    
    \begin{Lem}\label{relation}\cite{qu2025parametric}
        Let $P$ be a permutation over $\F_{2^n}$. Then for any $a \in \F_{2^n}^{*}$,
        \[
            W_{f_P}(a) = -S(aP(x)+x,n).
        \]
    \end{Lem}

    \section{Constructions of four-valued spectrum Boolean functions}\label{four-valued}
    This section introduces a class of four-valued spectrum Boolean functions, whose Walsh spectra imply the same high nonlinearity as semi-bent functions. The parameterization is selected as $P(x^2+x)$ with trinomial permutation $P(x)=cx+\tr_{m}^{2m}(x^{\frac{1}{5}(2^m-1)+1})$ over $\F_{2^{2m}}$. The permutation, along with other classes of permutations of the form $cx+\tr_{m}^{km}(x^s)$, was first studied by Li et al. in \cite{li2018permutation}.  

    \subsection{Techniques in Niho exponents}
    Niho exponents were originally introduced into the study of cross-correlation in \cite{niho1972multi} and yielded valuable results. In the decades that followed, this work also inspired the construction of bent functions, linear codes, and permutation polynomials \cite{li2019survey}. 

    Given positive integers $n,m$ with $n=2m$, a positive integer $s$ is termed a Niho exponent over $\F_{2^n}$ if it satisfies $s \equiv 2^j \pmod{2^m-1}$ for some $0 \leq j < n$. The case $j=0$ (yielding $s \equiv 1 \pmod{2^m-1}$) is called a normalized Niho exponent, whereas removing the restriction on residues being powers of 2 leads to the generalized Niho exponents. 
    
    Recall the notation $U_m = \{z \in \F_{2^n}~|~z \bar{z}=z^{2^m+1} = 1 \}$, which represents the unit circle within $\F_{2^{2m}}$. The following two well-known lemmas establish two well-known bijections concerning $U_m$ and serve as powerful tools for addressing Niho (classical or generalized) exponents. 
    \begin{Lem}\label{bi1}\cite{dobbertin2006niho}
        Let $n=2m$ and $U_m$ be the unit circle within $\F_{2^{n}}$. Then for any $x \in \F_{2^n}^{*}$, there exists a unique pair $(y,z) \in \F_{2^m}^{*} \times U_m$ such that $x=yz$. 
    \end{Lem}
    \begin{Lem}\label{bi2}\cite{lahtonen2002odd}
        Let $n=2m$ and $U_m$ be the unit circle within $\F_{2^{n}}$. For any $\omega \in \F_{2^n} \backslash \F_{2^m}$, then 
        \[
            U_m \backslash \{1\} = \left\{ \frac{x+\omega}{x+\bar{\omega}} \mid x \in \F_{2^m} \right\}.
        \]
    \end{Lem}

    Consider a polynomial of the form $x^r h(x^{2^m-1})$ where $r$ is a positive integer and $h(x)$ is a polynomial defined over $\F_{2^{2m}}$. It is observed that all the monomial exponents are of (generalized) Niho type within the same congruence class modulo $2^m-1$. This brings convenience to the calculation of the Walsh transform as established in Proposition \ref{nihopp}. 
    \begin{Prop}\label{nihopp}
        Let $n=2m$ and $G(x) = x^r h(x^{2^m-1})$, where  $r$ is a positive integer and   $h(x)$ is a polynomial over $\F_{2^n}$. Then the Walsh transform of $G$ at $(a,b) \in \F_{2^n}^{2}$ is given by
        \[
            W_G(a,b)=\sum_{z \in U_m}{S \left( \tr_m^n\left( bz^{r}h(\bar{z}^{2} )\right) y^{r} + \tr_m^n\left( az \right) y,m \right)}-2^m.
        \]
    \end{Prop}
    \begin{proof}
        By the definition of Walsh transform, for any pair $(a,b) \in \F_{2^n}^{2}$, 
        \begin{eqnarray*}
            W_G(a,b)
            &=& \sum_{x \in \F_{2^n}}{\chi_n(b G(x)+ax)} \\
            &=& \sum_{x \in \F_{2^n}^{*}}{\chi_n(b G(x)+ax)}+1 \\
            &=& \sum_{x \in \F_{2^n}^{*}}{\chi_n(b x^rh(x^{2^m-1})+ax)}+1.
        \end{eqnarray*}
        Decomposing each $x \in \F_{2^n}^{*}$ into the product of $y$ and $z$ with $(y,z) \in \F_{2^m}^{*} \times U_m$, noting that $y^{2^m-1}=1$ and $z^{2^m-1}=z^{-2}=\bar{z}^2$, it follows that
        \begin{eqnarray*}
            W_G(a,b)
            &=&\sum_{y \in \F_{2^m}^{*},z \in U_m}{\chi_n\left( by^{r}z^{r}h(\bar{z}^{2})+ayz \right)}+1 \\
            &=&\sum_{y \in \F_{2^m}}\sum_{z \in U_m}{\chi_n\left( by^{r}z^{r}h(\bar{z}^{2})+ayz \right)}-2^m \\
            &=&\sum_{z \in U_m}\sum_{y \in \F_{2^m}}{\chi_m\left( \tr_m^n\left( bz^{r}h(\bar{z}^{2}) \right) y^{r}+\left( az+\bar{a}\bar{z} \right)y \right)}-2^m, 
        \end{eqnarray*}
        where the last equality is owing to $\tr_n(x)=\tr_m \left( \tr_m^n(x) \right)$.
    \end{proof}
    In the case of normalized Niho exponent, Proposition \ref{nihothm} further reduces the problem to the solution counts of a specific equation on the unit circle. Relevant techniques have been successfully applied to the constructions of permutation polynomials \cite{tu2018new}, bent functions \cite{leander2006bent}, four-valued spectrum and six-valued spectrum functions \cite{li2013walsh,dobbertin2006niho}.
    \begin{Prop}\cite{dobbertin2006niho}\label{nihothm}
        Let $n=2m$ and $G(x) = x h(x^{2^m-1})$, where  $h(x)$ is a polynomial over $\F_{2^n}$. Then the Walsh transform of $G$ at $(a,b) \in \F_{2^n}^{2}$ is given by
        \[
            W_G(a,b) = 2^m\left( \#\{ z\in U_m \mid \tr_m^n\left( bzh(\bar{z}^{2}) \right)+az+\bar{a}\bar{z}=0 \}-1 \right).
        \]
    \end{Prop}

    \subsection{Constructions from permutations of the form $cx+\tr_{m}^{2m}(x^{\frac{1}{5}(2^m-1)+1})$}
    This subsection is concerned with the four-valued spectrum functions mentioned above and their value distributions. To facilitate the analysis, we first restate Lemmas \ref{projeq} and \ref{projeq2} that characterize the solution counts of specific quadratic equations.  
    \begin{Lem}\label{projeq}\cite{dobbertin2006niho}
        Let $n=2m$ and $k<m$. For any fixed $u \in \F_{2^n}^{*}$, the quadratic polynomials $Q_1(x)=ux^{2^k+1}+x^{2^k}+x+\bar{u}$ and $Q_2(x)=x^{2^k+1}+ux^{2^k}+\bar{u}x+1$ have either $0$, $1$, $2$, or $2^{\gcd(k,m)}+1$ distinct roots in $U_m$.
    \end{Lem}
    \begin{Lem}\label{projeq2}\cite{ding2023roots}
        Let $n,k$ be positive integers with $\gcd(k,n)=1$.
        Then for any pair $(a,b) \in \F_{2^n}^{*} \times \F_{2^n}$, the quadratic polynomial $Q(x)=x^{2^k+1}+bx+a$ has either $0$, $1$, or $3$ roots in $\F_{2^n}$, without considering multiplicity.  
    \end{Lem} 
    Indeed, Lemma \ref{projeq2} is a corollary of \cite[Lemma 2.2]{ding2023roots}, only simple coefficient substitution is needed to obtain it.
    \begin{Prop}\label{2diff}
        Let $n=2m$, $\gcd(k,n)=1$, and $u \in \F_{2^n}^{*}$. If the quadratic polynomials $Q_1(x)=ux^{2^k+1}+x^{2^k}+x+\bar{u}$ and $Q_2(x)=x^{2^k+1}+ux^{2^k}+\bar{u}x+1$
        have $2$ distinct roots in $U_m$, then one of these $2$ roots is a repeated root.
    \end{Prop}
    \begin{proof}
        Note that $\gcd(k,m)=1$ since $\gcd(k,n)=1$. According to Lemma \ref{projeq}, $Q_2(x)$ has either $0$, $1$, $2$, or $3$ distinct roots in $U_m$.
        
        Substituting $\tilde{x}=x+u$ into $Q_2(x)$, we define $\widetilde{Q_2}(\tilde{x})=\tilde{x}^{2^k+1}+(u^{2^k}+\bar{u})\tilde{x}+u\bar{u}+1$. Without considering multiplicity, $\widetilde{Q_2}(\tilde{x})$ and $Q_2(x)$ have the same number of roots in $\F_{2^n}$, which is either $0$, $1$, or $3$, as established by Lemma \ref{projeq2}. Consequently, we infer that $Q_2(x)$ has $2$ distinct roots in $U_m$ only when $Q_2(x)$ has a repeated root.

        The proof process for $Q_1(x)=ux^{2^k+1}+x^{2^k}+x+\bar{u}$ is similar and is omitted here.
    \end{proof}
    \begin{Prop}\label{fourbe}
        Let $n=2m$, $P(x)$ be a permutation for some $c \in \F_{2^n} \backslash \F_{2^m}$ in the form of
        \[
            cx+\tr_{m}^{n}(x^{\frac{2^k}{2^k+1}(2^m-1)+1})~(\text{resp.~} cx+\tr_{m}^{n}(x^{\frac{2^k}{2^k-1}(2^m-1)+1})),
        \]
        and $f_P$ be the balanced Boolean function parameterized by $P(x^2+x)$, where $k<m$ such that 
        \[
            \gcd(2^k+1,2^m+1)=1~ (\text{resp.~} \gcd(2^k-1,2^m+1)=1).
        \]
        Then the Walsh transform of $f_P$ at any $\gamma \in \F_{2^n}$ satisfies
        \[
            W_f(\gamma) \in \{2^m,0,-2^m,-2^{m+\gcd(k,m)}\}.
        \]
    \end{Prop}
    \begin{proof}
        Only the proof of case $P(x)=cx+\tr_{m}^{n}(x^{\frac{2^k}{2^k+1}(2^m-1)+1})$ would be given since the procedures are almost identical.

        Note that 
        \begin{eqnarray*}
            P(x) &=& cx+\tr_{m}^{n}(x^{\frac{2^k}{2^k+1}(2^m-1)+1}) \\
              &=& cx + x^{\frac{2^k}{2^k+1}(2^m-1)+1} + x^{\frac{1}{2^k+1}(2^m-1)+1} \\
              &=& xh(x^{2^m-1}),
        \end{eqnarray*}
         where $h(x)=c+x^{\frac{2^k}{2^k+1}}+x^{\frac{1}{2^k+1}}$.   By Lemma \ref{relation} and Proposition \ref{nihothm}, for any $\gamma \in \F_{2^n}$, we have 
        \begin{eqnarray*}
            W_{f_P}(\gamma)
            &=& -2^m ( \# \{ z\in U_m \mid \gamma z (c+\bar{z}^{\frac{2}{2^k+1}}+\bar{z}^{\frac{2^{k+1}}{2^{k}+1}} )+\bar{\gamma} \bar{z}  (\bar{c}+z^{\frac{2}{2^k+1}}+z^{\frac{2^{k+1}}{2^{k}+1}} )+z+\bar{z}=0  \}-1 ) \\
            &=& -2^m ( \# \{ z\in U_m \mid (\gamma+\bar{\gamma}) (z^{1-\frac{2}{2^k+1}}+z^{\frac{2}{2^k+1}-1} )+(\gamma c+1)z+(\bar{\gamma}\bar{c}+1)\bar{z}=0 \}-1 ).
        \end{eqnarray*}
        
        Our task now is to determine the number of solutions in $U_m$ to the equation  
        \begin{eqnarray}\label{fourpre}
            (\gamma+\bar{\gamma}) (z^{1-\frac{2}{2^k+1}}+z^{\frac{2}{2^k+1}-1} )+(\gamma c+1)z+(\bar{\gamma}\bar{c}+1)\bar{z}=0.
        \end{eqnarray}
        
        (\Rm{1}) In the case of $\gamma \in \F_{2^{m}}$, it is obvious that $\gamma+\bar{\gamma}=0$ and $\gamma c+1 \ne 0$. Then the Eq. (\ref{fourpre}) has exactly one solution, as it can be reduced to
        \begin{eqnarray*}
            (\gamma c+1)z+(\bar{\gamma}\bar{c}+1)\bar{z}=0
            \Leftrightarrow 
            z^2=\frac{\bar{\gamma}\bar{c}+1}{\gamma c+1}=\left( \gamma c+1 \right)^{2^m-1}.
        \end{eqnarray*}
        
        (\Rm{2}) In the case of $\gamma \in \F_{2^{n}} \backslash \F_{2^m}$, we first perform appropriate transformations on Eq. (\ref{fourpre}). Replacing $z$ by $z^{2^k+1}$, multiplying both sides by $z^{2^k+1}$, replacing $z^2$ by $z$, and dividing both sides by $\gamma+\bar{\gamma}$, we get 
        \begin{eqnarray}\label{fourfun}
            \frac{\gamma c+1}{\gamma+\bar{\gamma}}z^{2^k+1}+ z^{2^k}+ z+\frac{\bar{\gamma} \bar{c}+1}{\gamma+\bar{\gamma}}=0.
        \end{eqnarray}
        It immediately follows from Lemma \ref{projeq} that Eq. (\ref{fourfun}) has either $0$, $1$, $2$, or $2^{\gcd(k,m)}+1$ solutions in $U_m$ implying 
        \[
            W_{f_P}(\gamma) \in \{2^m,0,-2^m,-2^{m+\gcd(k,m)}\}.
        \]
    \end{proof}
    \begin{example}
        Consider the permutation $P(x)=cx+\tr_3^6(x^{15})$ with $c \in \F_{4} \backslash \F_{2}$ over $\F_{2^6}$. The corresponding Boolean function is $f_P(x)=\tr_6(x^{15}+x^{11}+cx)+1$ whose Walsh transform at any $\gamma \in \F_{2^6}$ such that $W_f(\gamma) \in \{8,0,-8,-16\}$.
    \end{example}
    \begin{Lem}\label{fourpp}\cite{li2018permutation}
        Let $n=2m$ with $m$ odd, $c \in \F_{2^2} \backslash \F_2$ and $P(x)=c x+\tr_{m}^{2m}(x^s)$. If the exponent $s=\frac{1}{5}(2^{m}-1)+1$, where $\frac{1}{5}$ is calculated modulo $2^m+1$, then trinomial $P(x)$ is a permutation over $\F_{2^n}$.
    \end{Lem}
    \begin{Thm}\label{Thm_fourvalued}
    Let $n=2m$ with $m$ odd, $P(x)=c x+\tr_{m}^{2m}(x^{\frac{1}{5}(2^{m}-1)+1}) \in \F_{2^{n}}[x]$, and $f_P$ be the balanced Boolean function parameterized by $P(x^2+x)$, where $c \in \F_{2^2} \backslash \F_2$. Then the Walsh spectrum of $f_P$ is 
    \[
    \begin{cases}
        0,&\text{occur~} (2^n-2^m+4)/2 \text{~times}\\
        2^m,&\text{occur~} (2^n-2^m-2)/3 \text{~times}\\
        -2^m,&\text{occur~} (2^m-2) \text{~times}\\
        -2^{m+1},&\text{occur~} (2^n-2^m+4)/6 \text{~times}
    \end{cases}
    \]
    \end{Thm}
    \begin{proof} 
        Observe that 
        \begin{eqnarray*}
            P(x) &=& c x+\tr_{m}^{2m}(x^{\frac{1}{5}(2^{m}-1)+1}) \\
            &=& cx + x^{\frac{1}{5}(2^{m}-1)+1} + x^{\frac{4}{5}(2^{m}-1)+1} \\
            &=& xh(x^{2^m-1}),
        \end{eqnarray*}
       where $h(x)=c+x^{\frac{1}{5}}+x^{\frac{4}{5}}$. Then, according to Proposition \ref{fourbe}, our problem reduces to computing the corresponding frequencies, since all possible values in the Walsh spectrum of $f_P$ can be derived from
        \begin{eqnarray*}
            W_{f_P}(\gamma)
            = -2^m( \#\{ z\in U_m \mid (\gamma c+1) z^{2^2+1}+  (\gamma+\bar{\gamma})z^{2^2}+ (\gamma+\bar{\gamma})z+ (\bar{\gamma} \bar{c}+1)=0 \}-1 ).
        \end{eqnarray*} 
            It is obvious that $W_f(\gamma)=0$ always holds if $\gamma \in \F_{2^m}$. When $\gamma \in \F_{2^n} \backslash \F_{2^m}$, it remains to compute the number of roots in $U_m$ to polynomial  
        \begin{eqnarray*}
            Q(z)=z^{2^2+1}+ \frac{\gamma+\bar{\gamma}}{\gamma c+1} z^{2^2}+ \frac{\gamma+\bar{\gamma}}{\gamma c+1} z+(\gamma c+1)^{2^m-1},
        \end{eqnarray*}
        which is either $0$, $1$, $2$, or $3$ by Lemma \ref{projeq}. 
        
        Let $N_i$ be the number of $\gamma \in \F_{2^n}$ such that $Q(x)$ has exactly $i$ roots in $U_m$. Combined with the relation (\ref{wr}), one can conclude that
        \begin{eqnarray}\label{foursys}
            \begin{cases}
                N_0+N_1+N_2+N_3=2^n,\\
                2^m N_0+0 \! \cdot \! N_1+(-2^m)N_2+(-2^{m+1})N_3=-2^n,\\
                2^n N_0+0^2 \! \cdot \! N_1+(-2^m)^2N_2+(-2^{m+1})^2N_3=2^{2n}.
            \end{cases}
        \end{eqnarray}
        It suffices to verify one of the values in $\{N_0, N_1, N_2, N_3\}$, as the equation system (\ref{foursys}) corresponds to a coefficient matrix of rank 3.
        
        Note that Proposition \ref{2diff} restricts $Q(x)$ to have two distinct roots if and only if $Q(x)$ has exactly a repeated root of multiplicity 2 and another root. Denote the formal derivative of $Q(x)$ by $Q^{\prime}(x)$. Then $Q(x)$ may have a repeated root if and only if $\gcd(Q(x),Q^{\prime}(x)) \ne 1$, or equivalently $\frac{\gamma+\bar{\gamma}}{\gamma c+1} \in U_m$ by the Euclid’s algorithm. When $Q(z)$ has a repeated root for some $\gamma$, then
        \begin{eqnarray*}
            Q(z)
            &=& z^{2^2+1}+ \frac{\gamma+\bar{\gamma}}{\gamma c+1} z^{2^2}+ \frac{\gamma+\bar{\gamma}}{\gamma c+1} z+(\gamma c+1)^{2^m-1}\\
            &=& \left(z+\frac{\gamma+\bar{\gamma}}{\gamma c+1}\right)\left(z^4+\frac{\gamma+\bar{\gamma}}{\gamma c+1}\right)+\left(\frac{\gamma+\bar{\gamma}}{\gamma c+1}\right)^{2}+(\gamma c+1)^{2^m-1}\\
            &=& \left(z+\frac{\gamma+\bar{\gamma}}{\gamma c+1}\right)\left(z^4+\frac{\gamma+\bar{\gamma}}{\gamma c+1}\right)+\frac{(\gamma+\bar{\gamma})^{2^m+1}}{(\gamma c+1)^2}+(\gamma c+1)^{2^m-1}\\
            &=& \left(z+\frac{\gamma+\bar{\gamma}}{\gamma c+1}\right)\left(z^4+\frac{\gamma+\bar{\gamma}}{\gamma c+1}\right).
        \end{eqnarray*}
        Therefore, $Q(x)$ has exactly $2$ solutions in $U_m$ if and only if $\frac{\gamma+\bar{\gamma}}{\gamma c+1} \in U_m$ and $\frac{\gamma+\bar{\gamma}}{\gamma c+1} \notin \F_{4}$, i.e.,
        \begin{eqnarray}\label{fourcir}
            \frac{\gamma+\bar{\gamma}}{\gamma c+1} \in U_m \backslash \{1,c,c^2\}.
        \end{eqnarray}

        Since $\{1,c\}$ is a basis of $\F_{2^{n}}$ over $\F_{2^{m}}$, for any $\gamma \in \F_{2^{n}} \backslash \F_{2^{m}}$, there exists a unique $(\gamma_1,\gamma_2) \in \F_{2^{m}} \times \F_{2^{m}}^{*}$ such that $\gamma=\gamma_1+\gamma_2 c$. Consequently, it can be concluded that: (\Rm{1}) $\frac{\gamma c+1}{\gamma+\bar{\gamma}}=1$ is equivalent to $(\gamma_1+\gamma_2)c=1$, which contradicts $\gamma_1+\gamma_2 \in \F_{2^{m}}$; (\Rm{2}) $\frac{\gamma c+1}{\gamma+\bar{\gamma}}=c^2$ is equivalent to $\gamma_1 c+1=0$, which contradicts $\gamma_1 \in \F_{2^{m}}$; (\Rm{3}) $\frac{\gamma c+1}{\gamma+\bar{\gamma}}=c$ is equivalent to $(\gamma_1,\gamma_2)=(0,1)$. From $\frac{\gamma c+1}{\gamma+\bar{\gamma}} \in U_m$, we establish that
        \begin{eqnarray*}
            &&\left(\frac{\gamma c+1}{\gamma+\bar{\gamma}}\right)^{2^{m}+1}+1=0 \\
            &\Rightarrow& \left( \gamma c+1 \right)^{2^{m}+1}+\left( \gamma+\bar{\gamma} \right)^{2^{m}+1}=0 \\
            &\Rightarrow& \gamma_2=\frac{\gamma_1^2+\gamma_1+1}{\gamma_1+1},~\gamma_1 \ne 1.
        \end{eqnarray*}
        
        Thus (\ref{fourcir}) holds if and only if
        \begin{eqnarray*}
            \gamma \in R=\left\{ \gamma_1+\gamma_2 c \; | \; \gamma_1 \in \F_{2^{m}} \backslash \F_{2},\gamma_2=\frac{\gamma_1^2+\gamma_1+1}{\gamma_1+1} \right\}.
        \end{eqnarray*}
        Clearly, $N_2=\# R =2^{m}-2$. Solving the equation system (\ref{foursys}), we get 
        \[
        \begin{cases}
            N_0=(2^{n}-2^{m}+4)/2, \\
            N_1=(2^{n}-2^{m}-2)/3, \\
            N_3=(2^{n}-2^{m}+4)/6.
        \end{cases}
        \]
        
        This proof can be finished here.
    \end{proof}
    
    \begin{example}\label{exm_fourvalued}
        When $m$ is assigned values of $3$, $5$, and $7$, the infinite class of permutations given in Lemma \ref{fourpp} takes the following explicit forms:
        \begin{eqnarray*}
            &&P_1(x)=cx+\tr_3^{6}(x^{15}) \in \F_{2^6}[x], \\
            &&P_2(x)=cx+\tr_5^{10}(x^{621}) \in \F_{2^{10}}[x], \\
            &&P_3(x)=cx+\tr_7^{14}(x^{3303}) \in \F_{2^{14}}[x],
        \end{eqnarray*}
        where $c \in \F_4 \backslash \F_2$. Then verified by MAGMA, the balanced Boolean functions $f_{P_i}(x)$ parameterized by $P_i(x^2+x)$ are all four-valued spectrum functions, whose spectral value distributions are consistent with Theorem \ref{Thm_fourvalued}.
        For further details about univariate representation and cryptographic criteria, refer to Table \ref{table1}, where $a^b$ means that $a$ occurs $b$ times in the Walsh spectrum.
        \begin{table}[H]
        \caption{Four-valued spectrum functions from $P(x^2+x)$ with $P(x) = cx+\tr_{m}^{2m}(x^{\frac{1}{5}(2^{m}-1)+1})$, $m$ odd, and $c \in \F_{4} \backslash \F_{2}$}
        \label{table1}	\centering
        \begin{threeparttable}
            \begin{tabular}{cccccc} 
                \toprule 
                $n$ & $f_P(x)=$ $\tr_n(R_n(x))+1$  & Walsh spectrum & $\mathcal{NL}(f)$ & $\deg(f)$ & AI$(f)$  \\
                \midrule  
                6 & $R_6(x) = x^{15}+x^{11}+cx$ & $\{0^{30},8^{18},-8^{6},-16^{10}\}$ & 24 & 4 & 3  \\
                10 & $R_{10}(x)$\tnote{1} & $\{0^{498},32^{330},-32^{30},-64^{166}\}$ & 480 & 6 & 5  \\
                14 & $R_{14}(x)$\tnote{2} & $\{0^{8130},128^{5418},-128^{126},-256^{2710}\}$  & 8064 & 8 & 7  \\
                \bottomrule 
            \end{tabular}
            \begin{tablenotes}
                \footnotesize
                \item[1] {\tiny $R_{10}(x) = x^{221} + x^{219} + x^{187} + x^{171} + x^{125} + x^{109} + x^{101} + x^{95} + x^{47} + x^{35} + cx$}
                \item[2] {\tiny $R_{14}(x) =x^{3445 }+ x^{2925 }+ x^{2795 }+ x^{2731 }+ x^{2683 }+ x^{2675 }+ x^{2671 }+ x^{2669 }+ x^{2541 }+ x^{2421 }+ x^{2415 }+ x^{1913 }+ x^{1909 }+x^{1907 }+ x^{1779 }+ x^{1715 }+ x^{1659 }+ x^{1655 }+ x^{1525 }+ x^{1461 }+ x^{1405 }+ x^{1401 }+ x^{1399 }+ x^{1271 }+ x^{1017 }+ x^{953 }+ x^{891 }+ x^{639 }+ x^{637 }+ x^{445 }+ x^{255 }+ x^{191 }+ x^{135 }+ x^{131 }+ cx$}
            \end{tablenotes}
        \end{threeparttable}
        \end{table}
        As shown in Table \ref{table1}, even Boolean functions parameterized by trinomials may be complex. Nevertheless, these four-valued spectrum functions still exhibit a high nonlinearity of $2^{2m-1}-2^m$ similar to the semi-bent functions. In the case of $n \leq 14$, the numerical results also indicate that these functions achieve the optimal algebraic immunity and have an algebraic degree of $m+1$ (the upper bound of semi-bent functions \cite{carlet2020boolean}).
    \end{example}

    \section{Constructions of plateaued functions}\label{plateaued functions}
    In this section, we present seven classes of plateaued functions including a class constructed from permutations of the form $x^{3 \cdot 2^{m}}+ax^{2 \cdot 2^{m}+1}+bx^{2^{m}+2}+cx^3$, three classes constructed from the so-called DO permutations, and three classes constructed from permutations of the form $cx+\tr_{m}^{km}(x^s)$. We first recall the explicit computation of the Walsh transform of quadratic functions, as detailed in Lemma \ref{qudratic_sum}.
    \begin{Lem}\cite{hou2007explicit}\label{qudratic_sum}
        Let $A(x)$ and $Q(x)$ be an affine function and a quadratic function over $\F_{2^n}$, respectively. Then for any pair $(a,b) \in \F_{2^n} \times \F_{2^n}^{*}$, quadratic Boolean function $\varphi_{a,b}=\tr_n(bQ(x)+ax)$ can be expressed as 
        \begin{eqnarray*}
            \varphi_{a,b}(x)
            =\tr_n\left(\sum_{i \in I} {a_ix^{2^{i}+1}}+A(x)\right),
        \end{eqnarray*}
        where $I \subset \left\{1,2,\cdots, \lfloor \frac{n}{2} \rfloor\right\}$ and
        \[
            a_i \in 
                \begin{cases}
                    \F_{2^n}^{*} , &\text{if~} i < \frac{n}{2},\\
                    \F_{2^n} \backslash \F_{2^\frac{n}{2}} ,  &\text{if~} i = \frac{n}{2}.
                \end{cases}
        \]
        Assume that $k= \max{I}$, define the associated bilinear mapping of $\varphi_{a,b}$ as 
        \begin{eqnarray*}
            B_{\varphi_{a,b}}(x,y)
            &=&\varphi_{a,b}(x+y)+\varphi_{a,b}(x)+\varphi_{a,b}(y) \\
            &=&\tr_n\left( x\sum_{i \in I}{\left( a_i^{2^k}y^{2^{k+i}}+a_i^{2^{k-i}}y^{2^{k-i}} \right)}^{2^{-k}} \right),
        \end{eqnarray*}
        and its kernel $V_{\varphi_{a,b}}$ is given by
        \begin{eqnarray}\label{gy}
            V_{\varphi_{a,b}}
            &=& \left\{y \in \F_{2^n} \mid B_{\varphi_{a,b}}(x,y)=0 \text{ for}~\forall x \in \F_{2^n} \right\} \nonumber \\
            &=& \left\{ y \in \F_{2^n}~\Big|~\sum_{i \in I}{\left( a_i^{2^k}y^{2^{k+i}}+a_i^{2^{k-i}}y^{2^{k-i}} \right)}=0 \right\}.
        \end{eqnarray}
        Then 	
        \begin{equation}\label{Eq-Quad}
            W_Q(a,b)=
            \begin{cases}
                \pm 2^{\frac{n+d}{2}}, & \text{if~} \varphi_{a,b} \text{~vanishes on~} V_{\varphi_{a,b}}, \\
                0, & \text{otherwise,} 
            \end{cases}
        \end{equation} 
        where $d$ is the dimension of $V_{\varphi_{a,b}}$ over $\F_2$.
    \end{Lem}
    
    The proofs of plateaued functions constructed by the last two largely rely on the quantity of roots for the polynomial $a^{2^m}x^{2^{2m}}+ax \in \F_{2^n}[x]$. This polynomial, where $a \neq 0$ and $m > 0$, has either $0$ or $2^{\gcd(2m,n)}$ nonzero solutions, as noted by Coulter \cite{coulter1999evaluation}.
    \begin{Lem}\cite{coulter1999evaluation}\label{binoeq}
        Let $n,m$ be positive integers, and $N$ be the number of nonzero solutions to the equation $a^{2^m}x^{2^{2m}}+ax=0 \text{~over~} \F_{2^n}$. For any $a \in \F_{2^n}^{*}$, we have
        \begin{eqnarray*}
            N= 
            \begin{cases}
              2^{d}-1,&\frac{n}{d} \text{~odd},\\
              2^{2 d}-1,&\frac{n}{d} \text{~even~and~} a=\alpha^{t(2^d+1)},\\
              0,&\text{~otherwise},
            \end{cases}
        \end{eqnarray*}
        where $\alpha$ is a primitive element of $\F_{2^n}$ and $d=\gcd(m,n)$.
    \end{Lem}
    
    \subsection{Constructions from $x^{3 \cdot 2^{m}}+ax^{2 \cdot 2^{m}+1}+bx^{2^{m}+2}+cx^3$}
    In \cite{tu2018new}, a well-known class of quadratic quadrinomials over $\F_{2^{n}}$ in the form 
    \[
         P(x)=x^{3 \cdot 2^{m}}+ax^{2 \cdot 2^{m}+1}+bx^{2^{m}+2}+cx^3
    \]
    was first studied as a permutation, where $m$ is odd. Further work promoted this result \cite{tu2019revisit,gupta2020several,kim2023completely} and revealed its excellent properties in boomerang uniformity \cite{tu2020class}. Here we concentrate on the subclass shown in Lemma \ref{qudpp} and obtain plateaued functions in Theorem \ref{Thm_plateaued}. 
    \begin{Lem}\label{qudpp}
    \cite{tu2018new}
        Let $m$ be an odd positive integer and $n= 2m$, $a, b, c\in \mathbb{F} _{2^{n}}$. Then
        \[
            P(x)=x^{3 \cdot 2^{m}}+ax^{2 \cdot 2^{m}+1}+bx^{2^{m}+2}+cx^3
        \]
        permutes $\mathbb{F}_{2^{n}}$ if one of the following three conditions is satisfied:\\
        (1) $c=1$, $a$ satisfies $\bar{\omega}\bar{a}+a+\omega\neq0$ and $b=\omega(a+1)+1$, where $\omega \in \F_{4} \backslash \F_{2}$;\\
        (2) $c= 1$, $b= a+ 1$ and $a+ \bar{a} + 1 \neq 0$;\\
        (3) $a,b,c \in\mathbb{F}_{2^{m}}$, $A=ab+c$, $B=1+a+b+c \neq 0$ and $C=a^2+b^2+ac+b$ satisfy either $C=0$ and $\tr_m\left(\frac{A}{B^2}\right)=1$, or $C=B^2$ and $\tr_m\left(\frac{A}{B^2}\right)=0$.
    \end{Lem}
    To complete the proof of plateaued functions in this subsection, we shall start with two lemmas to streamline the proof process. 
    Define the parameters $\epsilon_1,\epsilon_2,\epsilon_3,\epsilon_4$ as follows:
        \begin{eqnarray}\label{coffdef}
            \begin{cases*}
                \epsilon_1=\gamma+\bar{\gamma}+\gamma c+\bar{\gamma}\bar{c}+\gamma a+\gamma b+\bar{\gamma}\bar{a}+\bar{\gamma}\bar{b},\\
                \epsilon_2=\omega(\bar{\gamma}+\gamma c)+\bar{\omega}(\gamma+\bar{\gamma}\bar{c})+\omega(\gamma a+\bar{\gamma}\bar{b})+\bar{\omega}(\bar{\gamma}\bar{a}+\gamma b),\\
                \epsilon_3=\bar{\omega}(\bar{\gamma}+\gamma c)+\omega(\gamma+\bar{\gamma}\bar{c})+\bar{\omega}(\bar{\gamma}\bar{a}+\gamma b)+\omega(\gamma a+\bar{\gamma}\bar{b}),\\
                \epsilon_4=\bar{\gamma}+\gamma c+\gamma+\bar{\gamma}\bar{c}+\omega(\bar{\gamma}\bar{a}+\gamma b)+\bar{\omega}(\gamma a+\bar{\gamma}\bar{b}).
            \end{cases*}
        \end{eqnarray}
    Lemma \ref{coff} summarizes certain constraints on parameters $\epsilon_1,\epsilon_2,\epsilon_3,\epsilon_4$ for each condition stated in Lemma \ref{qudpp}, see \cite{tu2018new} for more details. Lemma \ref{linear} is related to the maximum number of roots of a linearized polynomial.
    \begin{Lem}\label{coff}\cite{tu2018new}
        Let $n= 2m$ with $m$ odd, $\omega \in \F_{2^2} \backslash \F_2$, $\gamma \in \mathbb{F} _{2^{n}}$, and $a, b, c$ be the coefficients stated in Lemma \ref{qudpp}. Parameters $\epsilon_1,\epsilon_2,\epsilon_3,\epsilon_4$ defined by Eq. (\ref{coffdef}) adhere to the following constraints:\\
        (\Rm{1}) If Condition (1) of Lemma \ref{qudpp} holds, then $\epsilon_1=\epsilon_2=\epsilon_3$, and thus $\epsilon_2^2+\epsilon_1\epsilon_3=0$;\\
        (\Rm{2}) If Condition (2) or (3) of Lemma \ref{qudpp}  holds, then there is 
        \[
            \tr_m\left( \frac{(\epsilon_2^2+\epsilon_1\epsilon_3)^3}{(\epsilon_1\epsilon_4+\epsilon_2\epsilon_3)^2\epsilon_1^2} \right)=0 
        \]
        holds true when $\epsilon_1=\gamma+\bar{\gamma} \ne 0$ or $\epsilon_1=(1+a+b+c)(\gamma+\bar{\gamma}) \ne 0$, respectively.
    \end{Lem}
    \begin{Lem}\cite{trachtenberg1970cross}\label{linear}
        Let $k$ be an integer with $\gcd(k,n)=1$ and $L$ be a linearized polynomial of the form 
        \[
            L(x) = c_0x+c_1x^{2^k}+c_2x^{2^{2k}}+\cdots + c_dx^{2^{dk}}\in\F_{2^n}[x]
        \]
        of degree $2^{dk}$. Then $L$ has at most $2^d$ zeros in $\F_{2^n}$.  
    \end{Lem}
    \begin{Thm}\label{Thm_plateaued}
        Let $n=2m$ with $m$ odd, $P(x) = x^{3\cdot2^{m}}+ax^{2\cdot2^{m}+1}+bx^{2^{m}+2}+cx^3$, and $f_P$ be the balanced Boolean function parameterized by $P(x^2+x)$, where $a,b,c$ satisfy the conditions in Lemma \ref{qudpp}. Then $f_P$ is $(m+1)$-plateaued if Condition (1) of Lemma \ref{qudpp} holds and $2$-plateaued (semi-bent) if Condition (2) or (3) of Lemma \ref{qudpp} holds.
    \end{Thm}
    \begin{proof}
        According to Lemma \ref{relation} and Proposition \ref{nihopp}, 
        \begin{eqnarray*}
            W_{f_P}(\gamma)
            &=&-\sum_{z \in U_m}\sum_{y \in \F_{2^{m}}}{\chi_m\left( \varphi_1(z) y^3+\varphi_2(z) y \right)}+2^{m},
        \end{eqnarray*}
        where
        \[
            \varphi_1(z)=u z^3+\bar{u}z^{-3}+v z+\bar{v}z^{-1},~
            \varphi_2(z)=z+z^{-1},
        \]
        and 
        \[
            u=\bar{\gamma}+\gamma c,~
            v=\gamma b+\bar{\gamma}\bar{a}.
        \]
        Let $\lambda(t)$ denote the bijection 
        \begin{eqnarray*}
            \lambda : \F_{2^{m}} & \to & U_m \backslash \{1\} \\
            t & \mapsto & \frac{t+\omega}{t+\bar{\omega}}
        \end{eqnarray*}
        where $\omega\in\F_4\backslash\F_2$.
        Then for any $z\in U_m\backslash\{1\}$, there exists a unique element $t\in\F_{2^{m}}$ such that $z=\lambda(t)$. Thus $\varphi_1(z)$ and $\varphi_2(z)$ are respectively simplified to
        \begin{eqnarray*}
            \varphi_1(\lambda(t))&=&u \left( \frac{t+\omega}{t+\omega^2} \right)^3+\bar{u}\left( \frac{t+\omega^2}{t+\omega} \right)^{3}+v \left( \frac{t+\omega}{t+\omega^2} \right)+\bar{v}\left( \frac{t+\omega^2}{t+\omega} \right)\\
            &=&\frac{(u+\bar{u})(t^6+1)+(u \omega+\bar{u} \omega^2)t^2+(u \omega^2+\bar{u}\omega)t^4}{(t^2+t+1)^3}+\frac{(v+\bar{v})t^2+v \omega^2+\bar{v} \omega}{t^2+t+1}\\
            &=&\frac{\varepsilon_1 t^6+\varepsilon_2 t^4+\varepsilon_3 t^2+\varepsilon_4}{(t^2+t+1)^3},
        \end{eqnarray*}
        and 
        \begin{eqnarray*}
            \varphi_2(\lambda(t))=\frac{1}{t^2+t+1},
        \end{eqnarray*}
        where 
        \[
        \begin{cases*}
            \varepsilon_1=u+\bar{u}+v+\bar{v}=\epsilon_1,\\
            \varepsilon_2=(\bar{u}+v)\omega+(u+\bar{v})\omega^2=\epsilon_2+\epsilon_1,\\
            \varepsilon_3=(u+v)\omega+(\bar{u}+\bar{v})\omega^2=\epsilon_3+\epsilon_1,\\
            \varepsilon_4=u+\bar{u}+v \omega^2+\bar{v}\omega=\epsilon_4+\epsilon,
        \end{cases*}
        \]
        and $\epsilon=v+\bar{v}$, $\epsilon_1,\epsilon_2,\epsilon_3,\epsilon_4$ are the same as defined in Lemma \ref{coff}. 
        Then we have
        \begin{eqnarray}\label{W}
            W_{f_P}(\gamma)&=&-\sum_{z \in U_m \backslash \{1\}}\sum_{y \in \F_{2^{m}}}{\chi_m\left( \varphi_1(z) y^3+\varphi_2(z) y \right)}-\sum_{y \in \F_{2^{m}}}{\chi_m(\varepsilon_1y^3)}+2^{m} \nonumber \\
            &=&-\sum_{t \in \F_{2^{m}}}\sum_{y \in \F_{2^{m}}}{\chi_m\left( \varphi_1(\lambda(t)) y^3+\varphi_2(\lambda(t)) y \right)}-\sum_{y \in \F_{2^{m}}}{\chi_m(\varepsilon_1y^3)}+2^{m}.
        \end{eqnarray}
        Note that $\gcd(3,2^{m}-1)=1$, therefore $y^3$ permutes $\F_{2^{m}}$ and then 
        \[
        \sum_{y \in \F_{2^{m}}}{\chi_m(\varepsilon_1y^3)}=
        \begin{cases}
            2^{m},& \text{if~} \varepsilon_1=\epsilon_1=0,\\
            0,&\text{otherwise}.
        \end{cases}
        \]
        The subsequent proof process depends on whether $\varepsilon_1$ is equal to $0$ or not. 
        
        (\Rm{1}) In the case of $\varepsilon_1=0$, (\ref{W}) is equivalent to
        \begin{eqnarray*}
            W_{f_P}(\gamma)&=&-\sum_{t \in \F_{2^{m}}}\sum_{y \in \F_{2^{m}}}{\chi_m\left( \varphi_1(\lambda(t)) y^3+\varphi_2(\lambda(t)) y \right)}~( \varphi_2(\lambda(t))y \to y) \\
            &=&-\sum_{t \in \F_{2^{m}}}\sum_{y \in \F_{2^{m}}}{\chi_m\left((\varepsilon_1 t^6+\varepsilon_2 t^4+\varepsilon_3 t^2+\varepsilon_4)y^3+y\right)} ~( t^2 \to t)\\
            &=&-\sum_{t \in \F_{2^{m}}}\sum_{y \in \F_{2^{m}}}{\chi_m\left((\varepsilon_1 t^3+\varepsilon_2 t^2+\varepsilon_3 t+\varepsilon_4)y^3+y\right)}\\
            &=&-\sum_{y \in \F_{2^{m}}}{\chi_m(\varepsilon_4y^3+y)}\sum_{t \in \F_{2^{m}}}{\chi_m\left((\varepsilon_2t^2+\varepsilon_3t)y^3\right)}\\
            &=&-\sum_{y \in \F_{2^{m}}}{\chi_m(\varepsilon_4 y^3 + y)}\sum_{t \in \F_{2^{m}}}{\chi_m\left(\left(\varepsilon_2 y^3+\varepsilon_3^2 y^6\right)t^2\right)}.
        \end{eqnarray*}
        
        If \textit{Condition (1)} holds, note that $\varepsilon_1=\varepsilon_2=\varepsilon_3=0$ due to $\varepsilon_1=\epsilon_1$ and $\epsilon_1=\epsilon_2=\epsilon_3$. Then for any $y\in \F_{2^{m}}$, we have $\sum_{t \in \F_{2^{m}}}{\chi_m\left(\left(\varepsilon_2 y^3+\varepsilon_3^2 y^6\right)t^2\right)}=2^m$. It is trival that $W_{f_P}(\gamma)=-2^{m}\sum_{y \in \F_{2^{m}}}{\chi_m(y)}=0$ when $\varepsilon_4 = 0$ holds true. And for any $\varepsilon_4 \ne 0$, we obtain
        \begin{eqnarray*}
            W_{f_P}(\gamma)&=&-2^{m}\sum_{y \in \F_{2^{m}}}{\chi_m(\varepsilon_4y^3+y)}~(y \to \varepsilon_4^{-\frac{1}{3}} y)\\
            &=&-2^{m}\sum_{y \in \F_{2^{m}}}{\chi_m\left(y^3+\varepsilon_4^{-\frac{1}{3}}y\right)}\\
            &=&
            \begin{cases}
                -\left(\frac{2}{m}\right) \chi_m(\theta^3+\theta)2^{\frac{3m+1}{2}} ,&\text{if}~x^{4}+x=\left(\varepsilon_4^{-\frac{1}{3}}+1\right)^2~\text{has a solution}~ \theta\in\F_{2^{m}}, \\
                0,&\text{otherwise},
            \end{cases}
        \end{eqnarray*}
        where the last equality is owing to Lemma \ref{Three_sum}.  
        
        Let $\gamma=\gamma_1+\gamma_2\omega$ with $\gamma_1,\gamma_2 \in \F_{2^m}$. Then $\gamma \bar{\omega}=\gamma_1\bar{\omega}+\gamma_2$, $\bar{\gamma}\omega=\gamma_1\omega+\gamma_2$. Noting that $\bar{\omega}\bar{a}+a+\omega \ne 0$, we obtain
        \begin{eqnarray*}
            \varepsilon_1 = 0 &\Leftrightarrow& u+\bar{u}+v+\bar{v} = 0\\
            &\Leftrightarrow& \gamma a+\gamma b+\bar{\gamma}\bar{a}+\bar{\gamma}\bar{b}=0\\
            &\Leftrightarrow& \gamma(a+b)+\bar{\gamma}(\bar{a}+\bar{b})=0\\
            &\Leftrightarrow& \gamma \bar{\omega}(a+1)+\bar{\gamma} \omega (\bar{a}+1)=0\\
            &\Leftrightarrow& \gamma_2(a+\bar{a})=\gamma_1+\gamma_1 \bar{\omega} a+\gamma_1 \omega \bar{a}\\
            &\Leftrightarrow& \gamma_2(a+\bar{a})\omega=\gamma_1(\bar{\omega}\bar{a}+a+\omega)\\
            &\Leftrightarrow& \gamma \in R=\left\{\gamma_1+\gamma_2 \omega ~\Big| ~\gamma_2 \in \F_{2^m},\gamma_1=\frac{(a+\bar{a})\omega}{\bar{\omega}\bar{a}+a+\omega}\gamma_2\right\}.
        \end{eqnarray*}
        Then it can be verified that
        \begin{eqnarray*}
            \varepsilon_4 &=& \gamma a+\bar{\gamma}\bar{b}\\
            &=& \gamma a+\bar{\gamma}\left( \bar{\omega} (\bar{a}+1)+1 \right)\\
            &=& \gamma a+\bar{\gamma}\bar{\omega}\bar{a}+\bar{\gamma}\omega\\
            &=& \left( \gamma_1+\gamma_2 \omega \right)a+\left( \gamma_1 \bar{\omega}+\gamma_2 \omega \right)\bar{a}+\gamma_1\omega+\gamma_2\\
            &=& \gamma_2+ \gamma_2 \omega\left( a+\bar{a} \right)+\gamma_1 \omega+\gamma_1 \bar{\omega} \bar{a}+\gamma_1 a\\
            &=& \gamma_2.
        \end{eqnarray*}
        That is to say, for any  triple $(a,b,c)$ satisfying \textit{Condition (1)}, precisely $2^{m}$ distinct values of $\gamma$ such that $\varepsilon_1=0$ hold, and half of these $\gamma$ additionally ensure $\tr_m(\varepsilon_4^{-1/3})=1$, when $\gamma_2$ ranges over $\F_{2^{m}}$. Thus in this case, $W_f(\gamma)\in\{ 0 ,\pm 2^{\frac{3m+1}{2}}\}$.
        
        If \textit{Condition (2)} or \textit{(3)} holds, then $\varepsilon_1=\gamma+\bar{\gamma} \text{~or~} (\gamma+\bar{\gamma})(1+a+b+c)$ holds accordingly. We claim that $\varepsilon_2 \varepsilon_3 \varepsilon_4=0$ if and only if $\varepsilon_2 =\varepsilon_3 =\varepsilon_4=0$. In fact, the premise $\varepsilon_1=0$ directly ensures $\gamma \in \F_{2^m}$ and further limits  
        \begin{eqnarray*}
            \begin{cases*}
                \varepsilon_2=\gamma(a+\bar{a}+1),\\
                \varepsilon_3=\gamma(a+\bar{a}+1),\\
                \varepsilon_4=\gamma(a+\bar{a}+1),
            \end{cases*}
            \text{~or~}
            \begin{cases*}
                \varepsilon_2=\gamma(a+b+c+1),\\
                \varepsilon_3=\gamma(a+b+c+1),\\
                \varepsilon_4=\gamma(a+b).
            \end{cases*}
        \end{eqnarray*}
        Thus $\varepsilon_2 \varepsilon_3 \varepsilon_4 \ne 0$ holds when $\gamma \in \F_{2^m}^{*}$, and we obtain
        \begin{eqnarray*}
            W_{f_P}(\gamma)&=&-2^{m}\sum_{\{y|\varepsilon_2 y^3+\varepsilon_3^2 y^6=0\}}{\chi_m(\varepsilon_4 y^3+y)}\\
            &=&-2^{m}\left(\chi_m(0)+\chi_m\left(\varepsilon_4\frac{\varepsilon_2}{\varepsilon_3^2}+\left( \frac{\varepsilon_2}{\varepsilon_3^2}\right)^{\frac{1}{3}} \right) \right)\\
            &=&0 \text{~or~} -2^{\frac{2m+2}{2}},
        \end{eqnarray*}
        since 
        \[
        \sum_{t \in \F_{2^{m}}}{\chi_m\left(\left(\varepsilon_2 y^3+\varepsilon_3^2 y^6\right)t^2\right)}
        =
        \begin{cases}
            2^{m},& \text{if~} \varepsilon_2 y^3+\varepsilon_3^2 y^6=0,\\
            0,&\text{otherwise}.
        \end{cases}
        \]
        Moreover, it is evident that when $\gamma=0$, the equation $W_{f_P}(\gamma)=0$ holds true.
        
        (\Rm{2}) In the case of $\varepsilon_1 \ne 0$, (\ref{W}) becomes 
        \begin{eqnarray}\label{ww}
            W_{f_P}(\gamma)&=&-\sum_{y \in \F_{2^{m}}}{\chi_m(\varepsilon_4 y^3+y)}\sum_{t \in \F_{2^{m}}}{\chi_m\left((\varepsilon_1 t^3+\varepsilon_2 t^2+\varepsilon_3 t)y^3\right)}+2^{m}~( t \to t+\frac{\varepsilon_2}{\varepsilon_1}) \nonumber \\
            &=&-\sum_{y \in \F_{2^{m}}^{*}}{\chi_m\left(\left(\varepsilon_4+\frac{\varepsilon_2\varepsilon_3}{\varepsilon_1}\right)y^3+y\right)}\sum_{t \in \F_{2^{m}}}{\chi_m\left(\left(\varepsilon_1t^3+\left(\frac{\varepsilon_2^2}{\varepsilon_1}+\varepsilon_3\right)t\right)y^3\right)}~( yt \to t) \nonumber \\
            &=&-\sum_{y \in \F_{2^{m}}^{*}}{\chi_m\left(\left(\varepsilon_4+\frac{\varepsilon_2\varepsilon_3}{\varepsilon_1}\right)y^3+y\right)}\sum_{t \in \F_{2^{m}}}{\chi_m\left(\varepsilon_1t^3+\left(\frac{\varepsilon_2^2}{\varepsilon_1}+\varepsilon_3\right)y^2t\right)}~( \varepsilon_1^{\frac{1}{3}}t \to t) \nonumber \\
            &=&-\sum_{y \in \F_{2^{m}}^{*}}{\chi_m\left(\left(\varepsilon_4+\frac{\varepsilon_2\varepsilon_3}{\varepsilon_1}\right)y^3+y\right)}\sum_{t \in \F_{2^{m}}}{\chi_m\left(t^3+\frac{\varepsilon_2^2+\varepsilon_1\varepsilon_3}{\varepsilon_1^{4/3}}y^2t\right)}.
        \end{eqnarray}
        If $\varepsilon_2^2+\varepsilon_1\varepsilon_3 \ne 0$, (\ref{ww}) can continue to do calculations as follows:
        \begin{eqnarray*}
            W_{f_P}(\gamma)&=&-\sum_{y \in \F_{2^{m}}^{*}}{\chi_m\left(\left(\varepsilon_4+\frac{\varepsilon_2\varepsilon_3}{\varepsilon_1}\right)^2y^6+y^2\right)}\sum_{t \in \F_{2^{m}}}{\chi_m\left(t^3+\frac{\varepsilon_2^2+\varepsilon_1\varepsilon_3}{\varepsilon_1^{4/3}}y^2t\right)}\\
            &=&-\chi_m(\theta^3+\theta)\left(\frac{2}{m}\right)2^{\frac{m+1}{2}-1}\sum_{\theta \in \F_{2^{m}}^{*}}{\chi_m\left(U^2V^3(\theta^4+\theta+1)^3+V(\theta^4+\theta+1)\right)}\\
            &=&-\left(\frac{2}{m}\right)2^{\frac{m-1}{2}}\sum_{\theta \in \F_{2^{m}}^{*}}{\psi(\theta)},
        \end{eqnarray*}
        where $\psi(\theta)=\chi_m\left(\Delta(\theta^4+\theta+1)^3+V(\theta^4+\theta+1)+\theta^3+\theta \right)$, $y^2/V=1+\theta+\theta^4$ with 
        \begin{eqnarray*}
            U=\frac{\varepsilon_1\varepsilon_4+\varepsilon_2\varepsilon_3}{\varepsilon_1},~V=\frac{\varepsilon_1^{4/3}}{\varepsilon_2^2+\varepsilon_1\varepsilon_3},~\Delta=U^2V^3.
        \end{eqnarray*}
        Thus, we have
        \begin{eqnarray*}
            \psi(\theta)
            &=&\chi_m\left(\Delta(\theta^8+\theta^2+1)(\theta^4+\theta+1)+V(\theta^4+\theta+1)+\theta^3+\theta \right)\\
            &=&\chi_m\left(\Delta+V+\left(1+V+\Delta+\Delta^{\frac{1}{2}}+V^{\frac{1}{4}}+\Delta^{\frac{1}{4}}+\Delta^{\frac{1}{8}}\right)\theta+\left( 1+\Delta^{\frac{1}{4}}+\Delta^{\frac{1}{2}}+\Delta\right)\theta^3+\Delta\theta^9 \right)\\
            &=&\chi_m\left(\Delta+V \right) \psi^{\prime}(\theta),
        \end{eqnarray*}
        with 
        \[
        \psi^{\prime}(\theta)=\chi_m\left(\left(1+V+\Delta+\Delta^{\frac{1}{2}}+V^{\frac{1}{4}}+\Delta^{\frac{1}{4}}+\Delta^{\frac{1}{8}}\right)\theta+\left( 1+\Delta^{\frac{1}{4}}+\Delta^{\frac{1}{2}}+\Delta\right)\theta^3+\Delta\theta^9 \right),
        \]
        it suffices to compute $\sum_{\theta \in \F_{2^{m}}}{\psi^{\prime}(\theta)}$. Denote the associated bilinear mapping of $\psi^{\prime}$ by 
        \begin{eqnarray*}
            B_{\psi^{\prime}}(x,y)&=&\psi^{\prime}(x+y)+\psi^{\prime}(x)+\psi^{\prime}(y)\\
            &=&\chi_m\left( \left( 1+\Delta^{\frac{1}{4}}+\Delta^{\frac{1}{2}}+\Delta\right) (xy^2+x^2y)+\Delta(xy^8+x^8y) \right)\\
            &=&\chi_m\left(\left(\Delta^8y^{4^3}+\left( 1+\Delta+\Delta^2+\Delta^4\right)^2y^{4^2}+\left( 1+\Delta+\Delta^2+\Delta^4\right)y^4+\Delta y \right)x^8\right)\\
            &=&\chi_m\left( \left(\Delta^8(y^{16}+y^{64})+(1+\Delta^2+\Delta^4)(y^4+y^{16})+\Delta(y+y^4)\right)x^8 \right),
        \end{eqnarray*}
        then by Lemma \ref{linear} the dimension $d$ of the kernel of 
        \begin{eqnarray*}
            L(y)=\Delta^8y^{4^3}+\left( 1+\Delta+\Delta^2+\Delta^4\right)^2y^{4^2}+\left( 1+\Delta+\Delta^2+\Delta^4\right)y^4+\Delta y
        \end{eqnarray*}
        is not greater than 3. Since $\frac{m+d}{2}$ is an integer and $m$ is odd, $d=1~\text{or}~3$.
        
        Let $Y=y^4+y$, the equation $L(y)=0$ becomes
        \begin{eqnarray*}
             \Delta^8Y^{16}+(1+\Delta^2+\Delta^4)Y^4+\Delta Y =0,
        \end{eqnarray*}
        obviously $y=0$ and $y=1$ are solutions of $L(y)=0$, then we have
        \begin{eqnarray}\label{lasteq}
            && \Delta^8Y^{15}+(1+\Delta^2+\Delta^4)Y^3+\Delta=0~(\Delta Y^3 \to Z) \nonumber\\
            &\Leftrightarrow& \Delta^3Z^5+\frac{1+\Delta^2+\Delta^4}{\Delta}Z+\Delta=0 \nonumber\\
            &\Leftrightarrow& Z^5+\left(1+\frac{1}{\Delta^2}+\frac{1}{\Delta^4}\right)Z+\frac{1}{\Delta^2}=0.
        \end{eqnarray}
        Next, we examine the corresponding Walsh transform values for each condition.
        
        If \textit{Condition (1)} holds, note that $\varepsilon_2^2+\varepsilon_1\varepsilon_3=\epsilon_2^2+\epsilon_1^2+\epsilon_1\epsilon_3+\epsilon_1^2=\epsilon_2^2+\epsilon_1\epsilon_3=0$, Eq. (\ref{ww}) becomes 
        \begin{eqnarray*}
            W_{f_P}(\gamma)&=&-\sum_{y \in \F_{2^{m}}^{*}}{\chi_m\left(\left(\varepsilon_4+\frac{\varepsilon_2\varepsilon_3}{\varepsilon_1}\right)y^3+y\right)}\sum_{t \in \F_{2^{m}}}{\chi_m\left(t^3\right)}\\
            &=&0,
        \end{eqnarray*}
        where the last equality due to $\sum_{t \in \F_{2^{m}}}{\chi_m\left(t^3\right)}=0$.
        
        If \textit{Condition (2)} or \textit{(3)} holds, we have $\varepsilon_2^2+\varepsilon_1\varepsilon_3=(\epsilon_2+\epsilon_1)^2+\epsilon_1(\epsilon_3+\epsilon_1)=\epsilon_2^2+\epsilon_1\epsilon_3$ and $(\varepsilon_1\varepsilon_4+\varepsilon_2\varepsilon_3)\varepsilon_1=\epsilon_1\epsilon_4+\epsilon_1\epsilon+\epsilon_2\epsilon_3+\epsilon_1^2+\epsilon_1(\epsilon_2+\epsilon_3)=(\epsilon_1\epsilon_4+\epsilon_2\epsilon_3)\epsilon_1$ since $\epsilon_1(\epsilon_1+\epsilon_2+\epsilon_3+\epsilon)=\epsilon_1\cdot 0$. Combined with Lemma \ref{coff}, we can obtain that
        \begin{eqnarray*}
            \tr_m\left(\frac{1}{\Delta^2}\right)&=&\tr_m\left(\frac{1}{\Delta}\right)\\
            &=&\tr_m\left(\frac{(\varepsilon_2^2+\varepsilon_1\varepsilon_3)^3}{(\varepsilon_1\varepsilon_4+\varepsilon_2\varepsilon_3)^2\varepsilon_1^2}\right)\\
            &=&\tr_m\left(\frac{(\epsilon_2^2+\epsilon_1\epsilon_3)^3}{(\epsilon_1\epsilon_4+\epsilon_2\epsilon_3)^2\epsilon_1^2}\right)\\
            &=&0.
        \end{eqnarray*}
        For Eq. (\ref{lasteq}), there exists an element $\tau$ such that $\frac{1}{\Delta^2}=\tau^2+\tau$, it suffices to consider 
        \begin{eqnarray}\label{le2}
            Z^5+(1+\tau+\tau^4)Z+\tau+\tau^2=0,
        \end{eqnarray}
        and $Z=\tau$ is exactly a solution of Eq. (\ref{le2}), which means the dimension $d$ is greater than $1$, and thus $d = 3$. Moreover, 
        \begin{eqnarray*}
            W_{f_P}(\gamma)=-\chi_m\left(\Delta+V \right)\left(\frac{2}{m}\right)2^{\frac{m-1}{2}}\sum_{\theta \in \F_{2^{m}}^{*}}{\psi(\theta)} \in \{0,\pm 2^{\frac{2m+2}{2}}\}.
        \end{eqnarray*}
    
        The proof of the theorem is now complete.
    \end{proof}
    \begin{example}
        The permutation $P(x)=x^{3 \cdot 2^m}+x^{2 \cdot 2^m+1}+x^3$ satisfies \textit{Condition (2)} in Lemma \ref{qudpp}, and the Boolean function parameterized by $P(x^2+x)$ is $2$-plateaued (semi-bent) on $\F_{2^{2m}}$ with main cryptographic criteria shown in Table \ref{table3}.
        \begin{table}[H]
		\caption{Semi-bent functions from $P_2(x^2+x)$ with $P_2(x)=x^{3 \cdot 2^m}+x^{2 \cdot 2^m+1}+x^3$ and $m$ odd} 
            \label{table3} \centering
		\begin{tabular}{ccccc}
			\toprule
			$n$ & $f_P(x)=$ $\tr_n(R_n(x))+1$ & $\mathcal{NL}(f)$ & $\deg(f)$ & AI$(f)$  \\
			\midrule
			6 & $R_6(x) = x^5$  & 24 & 2 & 2   \\
			10 & $R_{10}(x) =x^{73}+x^{13} + x^{11}$ &480 &3 & 3  \\
			14 & $R_{14}(x) = x^{593} + x^{553} + x^{85} + x^{53} + x^{43}$ & 8064 &4 & 4  \\
			\bottomrule
		\end{tabular}
	\end{table}
    \end{example}
    
    \subsection{Constructions from DO permutations}
    Plateaued functions in this subsection are derived from a special case of the Dembowski-Ostrom (DO) polynomial of the form $xL(x)$ with $L(x)$ linearized. For permutation property and other applications of DO polynomial, see \cite{blokhuis2001permutations,patarin1996hidden}.
    \begin{Lem}\label{xL1}\cite{blokhuis2001permutations}
        Let $n,m$ be positive integers, the following linearized polynomials $L(x) \in \F_{2^n}[x]$ yield permutation polynomials of the form $xL(x)$ over $\F_{2^n}$:\\
        (1) $L(x)=x^{2^m}$ with $\frac{n}{\gcd(n,m)}$ odd.\\
        (2) $L(x)=x^{2^m}+ax^{2^{n-m}}$ with $\frac{n}{\gcd(n,m)}$ odd and $a^{\frac{2^n-1}{2^{\gcd(n,m)}-1}} \ne 1$.\\
        (3) $L(x)=x^{2^{2m}}+a^{2^m+1}x^{2^{m}}+ax$ with $n=3m$ and $a^{\frac{2^n-1}{2^{m}-1}} \ne 1$.
    \end{Lem}
    \begin{Thm}\label{ex1}
        Let $n,m$ be positive integers with $\frac{n}{\gcd(n,m)}$ odd. Suppose that  $P(x)=xL(x)$ and $f_P$ is the balanced Boolean function parameterized by $P(x^2+x)$. Then $f_P$ is $\gcd(n,m)$-plateaued on $\F_{2^n}$ if one of the following two conditions is satisfied:\\
        (1) $L(x)=x^{2^m}$; \\
        (2) $L(x)=x^{2^m}+ax^{2^{n-m}}$ with $a^{\frac{2^n-1}{2^{\gcd(n,m)}-1}} \ne 1$.
    \end{Thm}
    \begin{proof}
        Since the proofs under either Condition \textit{(1)} or \textit{(2)} are similar, only the latter is given here. 
        
        In light of Lemmas \ref{relation} and \ref{binoeq}, 
        \begin{eqnarray*}
            W_{f_P}(\gamma)&=& -\sum_{x \in \F_{2^n}}{\chi_n(\gamma xL(x)+x)} \\
            &=& -\sum_{x \in \F_{2^n}}{\chi_n(\gamma x^{2^m+1}+ \gamma ax^{2^{n-m}+1}+x)} \\
            &=& 
            \begin{cases}
                \pm 2^{\frac{n+d}{2}},& \text{if~} \varphi \text{~vanishes on~} V_{\varphi},\\
                0, & \text{otherwise},
            \end{cases}
        \end{eqnarray*}
        where $\varphi=\tr_n(\gamma xL(x)+x)$, and $d$ is the dimension of 
        \[
            V_{\varphi}=\left\{y \in \F_{2^n} \big| \left( \gamma^{2^{m}}+\left( \gamma a \right)^{2^{2m}} \right)y^{2^{2m}}+\left( \gamma+\left( \gamma a \right)^{2^m} \right)y=0 \right\}.
        \]
        
        Noting that $\frac{n}{\gcd(n,m)}$ is odd, we now turn to the equation  
        \begin{eqnarray*}
            \left( \gamma^{2^{m}}+\left( \gamma a \right)^{2^{2m}} \right)y^{2^{2m}}+\left( \gamma+\left( \gamma a \right)^{2^m} \right)y=0,
        \end{eqnarray*}
        which has $2^{\gcd(n,m)}$ solutions if $\gamma+(\gamma a)^{2^m} \ne 0$. That is to say $d=\gcd(n,m)$, thus  
        \[
            W_{f_P}(\gamma) \in \left\{ 0,\pm 2^{\frac{n+\gcd(n,m)}{2}} \right\}.
        \]
        
        This proof has been finished.
    \end{proof}
    \begin{Rem}
        Consistent with the notation in Theorem \ref{ex1}. Let $n$ be any odd integer, and $m$ be a positive integer with $\gcd(n,m)=1$, other conditions remain the same. Then the parameterized Boolean function derived from $P(x^2+x)$ with $P(x)=x^{2^{m}}$ is $1$-plateaued (near-bent) on $\F_{2^n}$.
    \end{Rem}
    \begin{Thm}\label{ex2}
        Let $n=3m$, $P(x)=xL(x)$, and $f_P$ be the balanced Boolean function parameterized by $P(x^2+x)$, where $L(x)=x^{2^{2m}}+a^{2^m+1}x^{2^{m}}+ax$ with $a^{\frac{2^n-1}{2^{m}-1}} \ne 1$. Then $f_P$ is $m$-plateaued on $\F_{2^n}$.
    \end{Thm}
    \begin{proof}
        Similar to the proof process of Theorem \ref{ex1}, we obtain $W_{f_P}(\gamma) \in \{0,\pm 2^{\frac{n+d}{2}}\}$, where $\varphi=\tr_n(\gamma xL(x)+x)$, $b=\gamma^{2^{2m}}+( \gamma a^{2^m+1})^{2^m}$, and $d$ is the dimension of 
        \[
            V_{\varphi}=\left\{y \in \F_{2^n}~\big|~b y^{2^{m}}+b^{2^m}\left(y^{2^m}\right)^{2^{2m}}=0 \right\}.
        \]
        
        Note that $\gcd(2^m,2^n-1)=1$, hence $y^{2^m}$ permutes $\F_{2^n}$. Then equation $b y^{2^{m}}+b^{2^m}\left(y^{2^m}\right)^{2^{2m}}=0$ and $b y+b^{2^m}y^{2^{2m}}=0$ have the same number of solutions, which equals $2^m$, and thus $d=m$. 
        
        Now we can conclude that $f_P$ is $m$-plateaued on $\F_{2^n}$.
    \end{proof}
    \begin{example}
        Given a primitive element $\alpha$, all balanced Boolean functions parameterized by $P(x^2+x)$ with permutations $P(x)=xL(x)$ are 2-plateaued (semi-bent) on $\F_{2^6}$, where $L(x)$ equals $x^{2^2}$, $x^{2^2}+\alpha x^{2^{6-4}}$, or $x^{2^{2 \cdot 2}}+\alpha^{2^2+1} x^{2^{2}}+\alpha x$. Table \ref{table4} details their main cryptographic properties consistent with the claims presented in Theorems \ref{ex1} and \ref{ex2}.
        \begin{table}[H]
		\caption{Semi-bent functions from $P(x^2+x)$ with permutations $P(x)=xL(x)$ where $L(x)=x^{2^2},~x^{2^2}+\alpha x^{2^{6-4}}, \text{~or~} x^{2^{2 \cdot 2}}+\alpha^{2^2+1} x^{2^{2}}+\alpha x$ on $\F_{2^6}$} 
            \label{table4} \centering
		\begin{tabular}{ccccc}
			\toprule
			$L(x)$ & $f_P(x)=$ $\tr_n(R(x))+1$ & $\mathcal{NL}(f)$ & $\deg(f)$ & AI$(f)$  \\
			\midrule
			$x^{2^2}$ & $R(x) = x^{13}$  & 24 & 3 & 3   \\
			$x^{2^2}+\alpha x^{2^{6-4}}$ & $R(x) =\alpha^{53}x^{13} + \alpha^{34}x^{11} + \alpha^{15}x^7 + \alpha^{24}x^5 + \alpha^{19}x$ &24 &3 & 3  \\
			$x^{2^{2 \cdot 2}}+\alpha^{2^2+1} x^{2^{2}}+\alpha x$ & $R(x) = \alpha^{36}x^{13} + \alpha^{35}x^{11} + \alpha^{41}x^{7} + \alpha^{15}x^5 + \alpha^{45}x$ & 24 &3 & 3  \\
			\bottomrule
		\end{tabular}
	\end{table}
    \end{example}
    All semi-bent functions shown in Table \ref{table4} achieve the optimal algebraic immunity. Indeed, exhaustive search indicates that all balanced Boolean functions from Theorems \ref{ex1} and \ref{ex2} possess either optimal or suboptimal algebraic immunity for dimensions $n \leq 11$.
    
    \subsection{Constructions from permutation of the form $cx+\tr_{m}^{km}(x^s)$}
    Different from permutations with Niho exponents in Section \ref{four-valued}, here we choose the quadratic ones of the form $cx+\tr^{km}_{m}(x^s)$.
    \begin{Lem}\cite{li2018permutation}
        Let $m,k,i$ be positive integers. Then the following polynomials of the form $cx+\tr^{km}_{m}(x^s)$ permute $\F_{2^{km}}$.\\
        (1) $P=cx+\tr^{km}_{m}(x^{2^i(2^{m}+1)})$ where $c \in \F_{2^{2m}} \backslash \F_{2^{m}}$ and $k \equiv 2 \pmod{4}$;\\
        (2) $P=cx+\tr^{km}_{m}(x^{\frac{2^{2m}+2^{m}}{2}})$ where $c \in \F_{2^{m}} \backslash \F_2$ and $k$ odd;\\
        (3) $P=cx+\tr^{km}_{m}(x^{2(2^{im}+1)})$ where $c \in \F_{2^{m}}^{*}$, $c^{\frac{2^{m}-1}{3}} \ne 1$, $m \not \equiv 0 \pmod {3}$, $m$ even and $k$ odd.
    \end{Lem}
    \begin{Lem}\label{symmetry}
        Let $k,m,s$ be positive integers and $\gamma \in \F_{2^{km}}$. The trace functions $\tr_{km}(x)$ and $\tr_{m}^{km}(x)$ satisfy the following equation:
        \[
            \tr_{km}\left(\gamma \tr^{km}_{m}(x^s)\right)=\tr_{km}\left( \tr_{m}^{km}(\gamma) x^{s} \right).
        \]
    \end{Lem}
    \begin{proof}
        It is evident that
        \begin{eqnarray*}
            \tr_{km}\left(\gamma \tr^{km}_{m}(x^s)\right)
            &=& \tr_{km}\left(\gamma \left( x^{s}+x^{s2^m}+\cdots+x^{s2^{m(k-1)}} \right) \right)\\
            &=& \tr_{km}\left( \left( \gamma+\gamma^{\frac{1}{2^m}}+\cdots+\gamma^{\frac{1}{2^{m(k-1)}}} \right) x^{s} \right)\\
            &=& \tr_{km}\left( \tr_{m}^{km}(\gamma^{\frac{1}{2^{m(k-1)}}}) x^{s} \right)\\
            &=& \tr_{km}\left( \tr_{m}^{km}(\gamma) x^{s} \right).
        \end{eqnarray*}
    \end{proof}
    \begin{Thm}\label{pl2}
        Let $m,k,i$ be positive integers, $c \in \F_{2^{2m}} \backslash \F_{2^{m}}$, and $k \equiv 2 \pmod{4}$. Consider the balanced Boolean function $f_P$ derived from $P(x^2+x)$ with $P(x)=cx+\tr^{km}_{m}(x^{2^i(2^{m}+1)})$. Then $f_P$ is $2m$-plateaued on $\F_{2^{km}}$.
    \end{Thm}
    \begin{proof}
        According to Lemma \ref{symmetry}, we observe the fact that
        \begin{eqnarray*}
            \tr_{km}\left(\gamma \tr^{km}_{m}(x^{2^i(2^{m}+1)})\right)
            &=& \tr_{km}\left( \tr_{m}^{km}(\gamma) x^{2^i(2^{m}+1)} \right)\\
            &=& \tr_{km}\left( \tr_{m}^{km}(\gamma)^{\frac{1}{2^i}} x^{2^{m}+1} \right).
        \end{eqnarray*}
        Consequently, as established by Lemma \ref{relation}, the Walsh transform of the function $f$ evaluated at arbitrary $\gamma \in \F_{2^{km}}$ simplifies to
        \begin{eqnarray*}
            W_{f_P}(\gamma) &=& -\sum_{x \in \F_{2^{km}}}{\chi_{km}(\gamma P(x)+x)} \\
            &=& -\sum_{x \in \F_{2^{km}}}{\chi_{km}\left(\gamma \tr^{km}_{m}(x^{2^i(2^{m}+1)})+(\gamma c+1)x \right)}\\
            &=& -\sum_{x \in \F_{2^{km}}}{\chi_{km}\left(\tr_{m}^{km}(\gamma)^{\frac{1}{2^i}} x^{2^{m}+1}+(\gamma c+1)x \right)}.
        \end{eqnarray*}

        We claim that $\gamma c+1 \ne 0$ always holds for any $\gamma$ such that $\tr_{m}^{km}(\gamma)=0$. If the claim is not true, then $\gamma=c^{-1} \in \F_{2^{2m}} \backslash \F_{2^{m}}$, which implies $\gamma^{2^{2m}}+\gamma=0$ and $\gamma^{2^{m}}+\gamma \ne 0$. This leads to a contradiction
        \begin{eqnarray*}
            \tr_{m}^{km}(\gamma) &=& \gamma+\gamma^{2^{m}}+\cdots+\gamma^{2^{(k-1)m}}\\
            &=& \gamma+\gamma^{2^{m}} \ne 0.
        \end{eqnarray*}
        
        Therefore, we can make the conclusion that (\Rm{1}) if $\tr_{m}^{km}(\gamma)=0$, then $W_f(\gamma)=0$; (\Rm{2}) if $\tr_{m}^{km}(\gamma) \ne 0$, the bilinear form of 
        \[
            \tr_{km}\left(\tr_{m}^{km}(\gamma)^{\frac{1}{2^i}} x^{2^{m}+1}+(\gamma c+1)x\right)
        \]
        has a $2m$-dimensional kernel over $\F_{2}$, since $\gcd(2m,km)=2m$. Then we obtain
        \[
        W_{f_P}(\gamma) \in \left\{0,\pm 2^{\frac{km+2m}{2}}\right\}.
        \]
        
        Complete the proof here.
    \end{proof}
    \begin{Rem}
        Consistent with the notation in Theorem \ref{pl2}. Let $m=1$, and other conditions remain the same. Then parameterized Boolean function derived from $P(x^2+x)$ with $P(x)=cx+\tr_{k}(x^{3})$ is $2$-plateaued (semi-bent) on $\F_{2^k}$.
    \end{Rem}
    \begin{Thm}\label{pl3}
        Let $m,k,i$ be positive integers with $k$ odd, $c \in \F_{2^{m}} \backslash \F_{2}$. Consider the balanced Boolean function $f_P$ derived from $P(x^2+x)$ with $P(x)=cx+\tr^{km}_{m}(x^{\frac{2^{2m}+2^{m}}{2}})$. Then $f_P$ is $m$-plateaued on $\F_{2^{km}}$.
    \end{Thm}
    \begin{proof}
        Note that 
        \begin{eqnarray*}
            \tr_{km}\left(\gamma \tr^{km}_{m}(x^{\frac{2^{2m}+2^{m}}{2}})\right)
            &=& \tr_{km}\left( \tr_{m}^{km}(\gamma) x^{\frac{2^{m}(2^{m}+1)}{2}} \right)\\
            &=& \tr_{km}\left( \tr_{m}^{km}(\gamma)^{\frac{2}{2^{m}}} x^{2^{m}+1} \right),
        \end{eqnarray*}
        then we need to compute 
        \begin{eqnarray*}
            W_{f_P}(\gamma) &=& -\sum_{x \in \F_{2^{km}}}{\chi_{km}(\gamma P(x)+x)} \\
            &=& -\sum_{x \in \F_{2^{km}}}{\chi_{km}\left(\gamma \tr^{km}_{m}(x^{\frac{2^{2m}+2^{m}}{2}})+(\gamma c+1)x \right)}\\
            &=& -\sum_{x \in \F_{2^{km}}}{\chi_{km}\left(\tr_{m}^{km}(\gamma)^{\frac{2}{2^{m}}} x^{2^{m}+1}+(\gamma c+1)x \right)},
        \end{eqnarray*}
        divide it into two cases that $\gamma$ satisfies (\Rm{1}) $\tr_{m}^{km}(\gamma)=0$ or (\Rm{2}) $\tr_{m}^{km}(\gamma) \ne 0$.
        
        (\Rm{1}) In the case of $\tr_{m}^{km}(\gamma)=0$, it can be proved that $\gamma c+1 \ne 0$. Otherwise $\gamma =c^{-1} \in \F_{2^{m}} \backslash \F_{2}$, and then
        \begin{eqnarray*}
            \tr_{m}^{km}(\gamma) &=& \gamma+\gamma^{2^{m}}+\cdots+\gamma^{2^{(k-1)m}}\\
            &=& \underbrace{\gamma+\gamma+\cdots+\gamma}_{k \text{~times~with~} k \text{~odd}} \\
            &=& \gamma ~\ne~ 0,
        \end{eqnarray*}
        which contradicts the assumption. Thus we get
        \begin{eqnarray*}
            W_{f_P}(\gamma)=-\sum_{x \in \F_{2^{km}}}{\chi_{km}\left((\gamma c+1)x \right)}=0.
        \end{eqnarray*}

        (\Rm{2}) Consider the case of $\tr_{m}^{km}(\gamma) \ne 0$, then according to Lemma \ref{binoeq}, the bilinear form of 
        \[
            \tr_{km}\left(\tr_{m}^{km}(\gamma)^{\frac{2}{2^{m}}} x^{2^{m}+1}+(\gamma c+1)x\right)
        \]
        has a $m$-dimensional kernel over $\F_{2}$, since $\gcd(km,2m)=m$ and $\frac{km}{\gcd(km,m)}$ is odd, which implies 
        \[
            W_{f_P}(\gamma) \in \left\{ 0,\pm 2^{\frac{km+m}{2}}\right\}.
        \]

        The proof can be completed here. 
    \end{proof}
    \begin{Rem}
        Consistent with the notation in Theorem \ref{pl3}. Let $m=2$ and other conditions remain the same. Then parameterized Boolean function derived from $P(x^2+x)$ with $P(x)=cx+\tr_{2}^{2k}(x^{10})$ is $2$-plateaued (semi-bent) on $\F_{2^{2k}}$.
    \end{Rem}
    \begin{Thm}\label{pl4}
        Let $m,k,i$ be positive integers, where $k$ is odd, $m$ is even, and $m \not \equiv 0 \pmod{3}$. Consider the balanced Boolean function $f_P$ derived from $P(x^2+x)$ with $P(x)=cx+\tr^{km}_{m}(x^{2(2^{im}+1)})$, where $c \in \F_{2^{m}}^{*}$ such that $c^{\frac{2^{m}-1}{3}} \ne 1$. Then $f_P$ is $m\gcd(i,k)$-plateaued on $\F_{2^{km}}$.
    \end{Thm}
    \begin{proof}
        We claim that
        \begin{eqnarray*}
            \tr_{km}\left(\gamma \tr^{km}_{m}(x^{2(2^{im}+1)})\right)
            &=& \tr_{km}\left( \tr_{m}^{km}(\gamma) x^{2(2^{im}+1)} \right)\\
            &=& \tr_{km}\left( \tr_{m}^{km}(\gamma)^{\frac{1}{2}} x^{2^{im}+1} \right),
        \end{eqnarray*}
        and then it suffices to compute
        \begin{eqnarray*}
            W_{f_P}(\gamma) &=& -\sum_{x \in \F_{2^{km}}}{\chi_{km}\left(\tr_{m}^{km}(\gamma)^{\frac{1}{2}} x^{2^{im}+1}+(\gamma c+1)x \right)}.
        \end{eqnarray*}
        
         The proof process for the remaining part is similar to the previous theorem and is omitted here.
    \end{proof}
    \begin{Rem}
        Consistent with the notation in Theorem \ref{pl4}. Let $m=2$, $i$ be a positive integer with $\gcd(i,k)=1$, and other conditions remain the same. Then parameterized Boolean function derived from $P(x^2+x)$ with $P(x)=cx+\tr_{2}^{2k}(x^{2^{2i+1}+2})$ is $2$-plateaued (semi-bent) on $\F_{2^{2k}}$.
    \end{Rem}

    \section{Conclusions and further work}\label{conclusion}
    In this paper, the parametric construction approach was employed to construct few-valued spectrum Boolean functions from $2$-to-$1$ mappings of the form $P(x^2+x)$ where $P$ represents sparse permutation polynomials. By carefully selecting permutations with low degrees or Niho exponents, we constructed a new class of four-valued spectrum functions and seven plateaued functions, including four semi-bent subclasses and a near-bent subclass. 
     
    Our analysis demonstrates that the computability afforded by quadratic or Niho exponents enables the effective determination of Walsh spectra, directly revealing the nonlinear properties of constructed Boolean functions through their underlying algebraic structures. It is expected to obtain more valuable results by considering permutations and $2$-to-$1$ mappings with other specific properties, e.g., Dillon exponents.




\begin{thebibliography}{10}

    \bibitem{blokhuis2001permutations}
    Aart Blokhuis, Robert Coulter, Marie Henderson, and Christine O’Keefe.
    \newblock Permutations amongst the dembowski-ostrom polynomials.
    \newblock In {\em Finite Fields and Applications: Proceedings of The Fifth International Conference on Finite Fields and Applications, held at the University of Augsburg, Germany, August 2--6, 1999}, pages 37--42. Springer, 2001.
    
    \bibitem{canteaut2000binary}
    Anne Canteaut, Pascale Charpin, and Hans Dobbertin.
    \newblock Binary $m$-sequences with three-valued crosscorrelation: a proof of {Welch's} conjecture.
    \newblock {\em IEEE Transactions on Information Theory}, 46(1):4--8, 2000.
    
    \bibitem{cao2016further}
    Xiwang Cao, Hao Chen, and Sihem Mesnager.
    \newblock Further results on semi-bent functions in polynomial form.
    \newblock {\em Advances in Mathematics of Communications}, 10(4):725--741, 2016.
    
    \bibitem{carlet2015boolean}
    Claude Carlet.
    \newblock Boolean and vectorial plateaued functions and {APN} functions.
    \newblock {\em IEEE Transactions on Information Theory}, 61(11):6272--6289, 2015.
    
    \bibitem{carlet2020boolean}
    Claude Carlet.
    \newblock {\em Boolean functions for cryptography and coding theory}.
    \newblock Cambridge University Press, 2021.
    
    \bibitem{carlet2024Parameterization}
    Claude Carlet.
    \newblock Parameterization of {Boolean} functions by vectorial functions and associated constructions.
    \newblock {\em Advances in Mathematics of Communications}, 18(3):624--650, 2024.
    
    \bibitem{carlet2003plateaued}
    Claude Carlet and Emmanuel Prouff.
    \newblock On plateaued functions and their constructions.
    \newblock In {\em Fast Software Encryption: 10th International Workshop, FSE 2003, Lund, Sweden, February 24-26, 2003. Revised Papers 10}, pages 54--73. Springer, 2003.
    
    \bibitem{carlitz1979explicit}
    Leonard Carlitz.
    \newblock Explicit evaluation of certain exponential sums.
    \newblock {\em Mathematica Scandinavica}, 44(1):5--16, 1979.
    
    \bibitem{chee1995semi}
    Seongtaek Chee, Sangjin Lee, and Kwangjo Kim.
    \newblock Semi-bent functions.
    \newblock In {\em Advances in Cryptology—ASIACRYPT'94: 4th International Conferences on the Theory and Applications of Cryptology Wollongong, Australia, November 28--December 1, 1994 Proceedings 4}, pages 105--118. Springer, 1995.
    
    \bibitem{coulter1999evaluation}
    Robert Coulter.
    \newblock On the evaluation of a class of {Weil} sums in characteristic 2.
    \newblock {\em New Zealand Journal of Mathematics}, 28(2):171--184, 1999.
    
    \bibitem{cusick2017highly}
    Thomas Cusick.
    \newblock Highly nonlinear plateaued functions.
    \newblock {\em IET Information Security}, 11(2):78--81, 2017.
    
    \bibitem{ding2023roots}
    Zhiguo Ding and Michael Zieve.
    \newblock Roots of certain polynomials over finite fields.
    \newblock {\em Journal of Number Theory}, 252:157--176, 2023.
    
    \bibitem{dobbertin1998one}
    Hans Dobbertin.
    \newblock One-to-one highly nonlinear power functions on {GF($2^n$)}.
    \newblock {\em Applicable Algebra in Engineering, Communication and Computing}, 9(2):139--152, 1998.
    
    \bibitem{dobbertin2006niho}
    Hans Dobbertin, Patrick Felke, Tor Helleseth, and Petri Rosendahl.
    \newblock {Niho} type cross-correlation functions via {Dickson} polynomials and {Kloosterman} sums.
    \newblock {\em IEEE Transactions on Information Theory}, 52(2):613--627, 2006.
    
    \bibitem{gupta2020several}
    Rohit Gupta.
    \newblock Several new permutation quadrinomials over finite fields of odd characteristic.
    \newblock {\em Designs, Codes and Cryptography}, 88(1):223--239, 2020.
    
    \bibitem{helleseth2005new}
    Tor Helleseth and Petri Rosendahl.
    \newblock New pairs of $m$-sequences with 4-level cross-correlation.
    \newblock {\em Finite Fields and Their Applications}, 11(4):674--683, 2005.
    
    \bibitem{hodvzic2020characterization}
    Samir Hod{\v{z}}i{\'c}, Peter Horak, and Enes Pasalic.
    \newblock Characterization of basic 5-value spectrum functions through {Walsh-Hadamard} transform.
    \newblock {\em IEEE Transactions on Information Theory}, 67(2):1038--1053, 2020.
    
    \bibitem{hodvzic2020general}
    Samir Hod{\v{z}}i{\'c}, Enes Pasalic, and Yongzhuang Wei.
    \newblock A general framework for secondary constructions of bent and plateaued functions.
    \newblock {\em Designs, Codes and Cryptography}, 88:2007--2035, 2020.
    
    \bibitem{hodvzic2019designing}
    Samir Hod{\v{z}}i{\'c}, Enes Pasalic, Yongzhuang Wei, and Fengrong Zhang.
    \newblock Designing plateaued {Boolean} functions in spectral domain and their classification.
    \newblock {\em IEEE Transactions on Information Theory}, 65(9):5865--5879, 2019.
    
    \bibitem{hodvzic2019generic}
    Samir Hod{\v{z}}i{\'c}, Enes Pasalic, and Weiguo Zhang.
    \newblock Generic constructions of five-valued spectra {Boolean} functions.
    \newblock {\em IEEE Transactions on Information Theory}, 65(11):7554--7565, 2019.
    
    \bibitem{hou2007explicit}
    Xiang-Dong Hou.
    \newblock Explicit evaluation of certain exponential sums of binary quadratic functions.
    \newblock {\em Finite Fields and Their Applications}, 13(4):843--868, 2007.
    
    \bibitem{kim2023completely}
    Kwang Kim, Sihem Mesnager, Chung Kim, and Myong Jo.
    \newblock Completely characterizing a class of permutation quadrinomials.
    \newblock {\em Finite Fields and Their Applications}, 87:102155, 2023.
    
    \bibitem{lahtonen2002odd}
    J.~Lahtonen.
    \newblock On the odd and the aperiodic correlation properties of the {Kasami} sequences.
    \newblock {\em IEEE Transactions on Information Theory}, 41(5):1506--1508, 2002.
    
    \bibitem{leander2006bent}
    Gregor Leander and Alexander Kholosha.
    \newblock Bent functions with $2^r$ {Niho} exponents.
    \newblock {\em IEEE Transactions on Information Theory}, 52(12):5529--5532, 2006.
    
    \bibitem{li2020class}
    Fengwei Li, Yansheng Wu, and Qin Yue.
    \newblock A class of functions with low-valued {Walsh} spectrum.
    \newblock {\em Discrete Applied Mathematics}, 279:92--105, 2020.
    
    \bibitem{li2018permutation}
    Kangquan Li, Longjiang Qu, Xi~Chen, and Chao Li.
    \newblock Permutation polynomials of the form $cx+ \mathrm{Tr}_{q^l/q} (x^a)$ and permutation trinomials over finite fields with even characteristic.
    \newblock {\em Cryptography and Communications}, 10:531--554, 2018.
    
    \bibitem{li2013walsh}
    Nian Li, Tor Helleseth, Alexander Kholosha, and Xiaohu Tang.
    \newblock On the {Walsh} transform of a class of functions from {Niho} exponents.
    \newblock {\em IEEE Transactions on Information Theory}, 59(7):4662--4667, 2013.
    
    \bibitem{li2019survey}
    Nian Li and Xiangyong Zeng.
    \newblock A survey on the applications of {Niho} exponents.
    \newblock {\em Cryptography and Communications}, 11:509--548, 2019.
    
    \bibitem{li2022class}
    Zhen Li, Haode Yan, and Dongchun Han.
    \newblock A class of power functions with four-valued {Walsh} transform and related cyclic codes.
    \newblock {\em Finite Fields and Their Applications}, 83:102078, 2022.
    
    \bibitem{lidl1997finite}
    Rudolf Lidl and Harald Niederreiter.
    \newblock {\em Finite fields}.
    \newblock Number~20. Cambridge University Press, 1997.
    
    \bibitem{logachev2010recursive}
    A.O. Logachev.
    \newblock On a recursive class of plateaued {Boolean} functions.
    \newblock {\em Discrete Mathematics and Applications}, 20(5):537--551, 2010.
    
    \bibitem{luo2016binary}
    Jinquan Luo.
    \newblock Binary sequences with three-valued cross correlations of different lengths.
    \newblock {\em IEEE Transactions on Information Theory}, 62(12):7532--7537, 2016.
    
    \bibitem{luo2008cyclic}
    Jinquan Luo and Keqin Feng.
    \newblock Cyclic codes and sequences from generalized {Coulter--Matthews} function.
    \newblock {\em IEEE Transactions on Information Theory}, 54(12):5345--5353, 2008.
    
    \bibitem{luo2008weight}
    Jinquan Luo and Keqin Feng.
    \newblock On the weight distributions of two classes of cyclic codes.
    \newblock {\em IEEE Transactions on Information Theory}, 54(12):5332--5344, 2008.
    
    \bibitem{maitra2002cryptographically}
    Subhamoy Maitra and Palash Sarkar.
    \newblock Cryptographically significant {Boolean} functions with five valued {Walsh} spectra.
    \newblock {\em Theoretical Computer Science}, 276(1-2):133--146, 2002.
    
    \bibitem{matsui1993linear}
    Mitsuru Matsui.
    \newblock Linear cryptanalysis method for {DES} cipher.
    \newblock In {\em Workshop on the Theory and Application of Cryptographic Techniques}, pages 386--397. Springer, 1993.
    
    \bibitem{meier1988fast}
    Willi Meier and Othmar Staffelbach.
    \newblock Fast correlation attacks on stream ciphers.
    \newblock In {\em Workshop on the Theory and Application of Cryptographic Techniques}, pages 301--314. Springer, 1988.
    
    \bibitem{mesnager2014semi}
    Sihem Mesnager.
    \newblock On semi-bent functions and related plateaued functions over the {Galois} field $\mathbb{F}_{2^n}$.
    \newblock {\em Open Problems in Mathematics and Computational Science}, pages 243--273, 2014.
    
    \bibitem{niho1972multi}
    Yoji Niho.
    \newblock {\em Multi-valued cross-correlation functions between two maximal linear recursive sequences}.
    \newblock PhD thesis, University of Southern California, 1972.
    
    \bibitem{patarin1996hidden}
    Jacques Patarin.
    \newblock Hidden fields equations ({HFE}) and isomorphisms of polynomials ({IP}): {Two} new families of asymmetric algorithms.
    \newblock In {\em International Conference on the Theory and Application of Cryptographic Techniques}, pages 33--48. Springer, 1996.
    
    \bibitem{qu2025parametric}
    Longjiang Qu, Qiancheng Zhang, and Kangquan Li.
    \newblock Parametric construction approach of balanced {Boolean} functions from two-to-one mappings.
    \newblock {\em Journal of Cryptology}, 38(3):1--42, 2025.
    
    \bibitem{rothaus1976bent}
    Oscar Rothaus.
    \newblock On “bent” functions.
    \newblock {\em Journal of Combinatorial Theory, Series A}, 20(3):300--305, 1976.
    
    \bibitem{su2022constructions}
    Sihong Su, Bingxin Wang, and Jingjing Li.
    \newblock On the constructions of resilient {Boolean} functions with five-valued {Walsh} spectra and resilient semi-bent functions.
    \newblock {\em Discrete Applied Mathematics}, 309:1--12, 2022.
    
    \bibitem{sun2015boolean}
    Zhiqiang Sun and Lei Hu.
    \newblock Boolean functions with four-valued {Walsh} spectra.
    \newblock {\em Journal of Systems Science and Complexity}, 28(3):743--754, 2015.
    
    \bibitem{sun2017several}
    Zhiqiang Sun and Lei Hu.
    \newblock Several classes of {Boolean} functions with four-valued {Walsh} spectra.
    \newblock {\em International Journal of Foundations of Computer Science}, 28(04):357--377, 2017.
    
    \bibitem{trachtenberg1970cross}
    H.M. Trachtenberg.
    \newblock {\em On the cross-correlation functions of maximal recurring sequences}.
    \newblock PhD thesis, University of Southern California, 1970.
    
    \bibitem{tu2020class}
    Ziran Tu, Nian Li, Xiangyong Zeng, and Junchao Zhou.
    \newblock A class of quadrinomial permutations with boomerang uniformity four.
    \newblock {\em IEEE Transactions on Information Theory}, 66(6):3753--3765, 2020.
    
    \bibitem{tu2019revisit}
    Ziran Tu, Xianping Liu, and Xiangyong Zeng.
    \newblock A revisit to a class of permutation quadrinomials.
    \newblock {\em Finite Fields and Their Applications}, 59:57--85, 2019.
    
    \bibitem{tu2018new}
    Ziran Tu, Xiangyong Zeng, and Tor Helleseth.
    \newblock New permutation quadrinomials over $\mathbb{F}_{2^{2m}}$.
    \newblock {\em Finite Fields and Their Applications}, 50:304--318, 2018.
    
    \bibitem{tu2011boolean}
    Ziran Tu, Dabin Zheng, Xiangyong Zeng, and Lei Hu.
    \newblock Boolean functions with two distinct {Walsh} coefficients.
    \newblock {\em Applicable Algebra in Engineering, Communication and Computing}, 22:359--366, 2011.
    
    \bibitem{xie2009class}
    Yonghong Xie, Lei Hu, Wenfeng Jiang, and Xiangyong Zeng.
    \newblock A class of {Boolean} functions with four-valued {Walsh} spectra.
    \newblock In {\em 2009 15th Asia-Pacific Conference on Communications}, pages 880--883. IEEE, 2009.
    
    \bibitem{xu2017several}
    Guangkui Xu, Xiwang Cao, and Shanding Xu.
    \newblock Several classes of {Boolean} functions with few {Walsh} transform values.
    \newblock {\em Applicable Algebra in Engineering, Communication and Computing}, 28:155--176, 2017.
    
    \bibitem{zheng1999plateaued}
    Yuliang Zheng and Xian-Mo Zhang.
    \newblock Plateaued functions.
    \newblock In {\em Information and Communication Security: Second International Conference, ICICS’99, Sydney, Australia, November 9-11, 1999. Proceedings 2}, pages 284--300. Springer, 1999.

    \end{thebibliography}
 \end{document}